\documentclass{article}
\usepackage{romp}

\usepackage{latexsym}
\usepackage{fancyhdr}
\usepackage[mathscr]{eucal}
\usepackage{amsmath}
\usepackage{mathrsfs}
\usepackage{amsthm}
\usepackage{amsfonts} 
\usepackage{amssymb} 
\usepackage{amscd}
\usepackage{bbm}
\usepackage{graphicx}
\usepackage{graphics}

\newcommand{\ii}{\mathrm{i}}
\newcommand{\cH}{\mathcal{H}}
\newcommand{\ran}{\mathrm{ran}}
\newcommand{\ud}{\mathrm{d}}

\title{ On point interactions realised as Ter-Martirosyan--Skornyakov Hamiltonians }
\author{ Alessandro Michelangeli, Andrea Ottolini \thanks{ \;\; Supported by a 2014-2015 ``\emph{INdAM grant Progetto Giovani}'', by the 2014-2017 MIUR-FIR grant ``\emph{Cond-Math: Condensed Matter and Mathematical Physics}'' code RBFR13WAET, and by a visiting research fellowship at the International Center for Mathematical Research CIRM, Trento 2015.}\\ SISSA -- International School for Advanced Studies \\
				Via Bonomea 265, 34136 Trieste (Italy) \\ e-mail: alemiche@sissa.it, aottolini@sissa.it \\[2ex]
	}

\begin{document}

\maketitle
\begin{abstract}
      For quantum systems of zero-range interaction we discuss the mathematical scheme within which modelling the two-body interaction by means of the physically relevant ultra-violet asymptotics known as the ``Ter-Martirosyan--Skornyakov condition'' gives rise to a self-adjoint realisation of the corresponding Hamiltonian. This is done within the self-adjoint extension scheme of Kre{\u\i}n, Vi\v{s}ik, and Birman. We show that the Ter-Martirosyan--Skornyakov asymptotics is a condition of self-adjointness only when is imposed in suitable functional spaces, and not just as a point-wise asymptotics, and we discuss the consequences of this fact on a model of two identical fermions and a third particle of different nature.
\end{abstract}

\noindent
{\bf Keywords:} Point interactions, self-adjoint extensions, Kre{\u\i}n-Vi\v{s}ik-Birman theory, Ter-Martirosyan--Skornyakov operators.

\section{Introduction}

According to a nomenclature that has emerged in various physical and mathematical contexts,  one refers to the so-called Ter-Martirosyan--Skornyakov (henceforth  TMS)  operators  as a distinguished class of quantum Hamiltonians for systems of non-relativistic particles with two-body ``\emph{zero-range}'' (or ``\emph{contact}'', or ``\emph{point}'') interaction. This terminology stems from early works in nuclear physics, where it was the nucleon-nucleon coupling to be initially modelled as a ``contact'' interaction. Nowadays the typical experimental realisation is that of ultra-cold atom systems where, by Feshbach resonance methods, the two-body scattering length is tuned to a magnitude that exceeds by many orders its nominal value, and the effective range of the interaction shrinks correspondingly to a very small scale, so that to an extremely good approximation the interaction can be considered to be of  infinite scattering length and/or zero range. In Section \ref{sec:history_TMS} we will provide a more diffuse context and references.

Informally speaking,  TMS Hamiltonians are qualified by the two characteristics of acting as the $N$-body $d$-dimensional \emph{free} Hamiltonian on functions that are supported \emph{away} from the ``coincidence hyperplanes'' $\{x_i=x_j\}$, and of having a domain that consists of square-integrable functions $\Psi(x_1,\dots,x_N)$, possibly with  fermionic or bosonic exchange symmetry, which satisfy specific asymptotics when $|x_i-x_j|\to 0$ for some or for all particle couples $i,j$. This models an interaction supported only on the hyperplanes $\{x_i=x_j\}$. It is customary to refer to this ultra-violet asymptotics as the ``\emph{TMS condition}''.

The explicit form for the TMS condition has various versions (see, e.g., \eqref{eq:TMS-generic}, \eqref{eq:BP-generic}, \eqref{eq:Berezin-Faddeev-2}, \eqref{eq:TMS_recovered_1+1}, \eqref{eq:TMS_cond_asymptotics_1}, or \eqref{eq:TMS_cond_asymptotics_2+1} below), all essentially equivalent to each other. Noticeably, such asymptotics emerge from different contexts and languages: on the one side the physical heuristics for an ``effective'' low-energy two-body scattering due to an interaction of very short range, on the other side the mathematical theory of self-adjoint extensions of symmetric operators on Hilbert space.
This is a fascinating history of reciprocal influence and mutual inspiration between ``early days'' nuclear physics, modern condensed matter physics, and mathematical operator theory and self-adjoint extension theory. Section \ref{sec:history_TMS} below will partially survey it.

TMS Hamiltonians represent the modern operator-theoretic approach to multi-particle quantum systems with two-body point interaction, and have an intimate connection to the alternative approach based on energy quadratic forms. 
They arise as natural effective models, based on stringent physical heuristics on the behaviour of the many-body wave-function when \emph{two} particles come on top of each other. In many circumstances, however, which depend essentially on the mass of the particles and on possible additional symmetries of the system, a formal TMS Hamiltonian fails to be self-adjoint and each of its self-adjoint extensions accounts for a different behaviour of the system when \emph{three} particles get closer and closer to the same point. The question is then to identify these extensions (if more than one) and to study their stability and spectral properties.

This is even more so since quantum systems with zero-range interactions may exhibit two somewhat exotic phenomena, as compared with the case of ordinary finite range potentials: the so-called ``\emph{Thomas effect}'', namely the emergence of an infinite discrete sequence of bound states with negative energy diverging to $-\infty$ and eigenfunction collapsing onto the barycentre, and the ``\emph{Efimov effect}'', that consists of an infinite sequence of bound states with negative energy arbitrarily close to zero and eigenfunctions extending on a larger and larger spatial scale, with a non-square-integrable limit. The TMS Hamiltonians account for such phenomena as well.


In this work we put the emphasis on the mathematical scheme within which a two-body interaction modelled by means of  a TMS condition (a \emph{physical} requirement) does correspond to a self-adjoint Hamiltonian (a \emph{mathematical} constraint that ensures the well-posedness and the correct interpretation of the quantum model).

Here is how the general problem is posed, as we shall elaborate further in the historical review of  Section \ref{sec:history_TMS} and then in the concrete settings of Sections \ref{sec:2body-point} and \ref{sec:2+1}
\begin{itemize}
 \item[1.] One starts (mostly in $d=3$ dimensions) with the operator $\mathring{H}$ obtained by \emph{restricting} the $N$-body free Hamiltonian to the regular wave-functions that are supported away from the coincidence hyperplanes $\Gamma_{ij}:=\{x_i=x_j\}$. Additional partial or global symmetries for $\mathring{H}$ are possible, such as rotational symmetries, or exchange bosonic or fermionic symmetries for some or all of the $N$ particles. $\mathring{H}$ is clearly densely defined, symmetric, and positive, and any its self-adjoint extension is naturally interpreted as a model for an interaction supported at $\Gamma_{ij}$.
%
 \item[2.] The self-adjoint extensions of $\mathring{H}$ are restrictions of $\mathring{H}^*$. Thus, one first characterises the domain and the action of $\mathring{H}^*$ and then one selects a special class of restrictions of $\mathring{H}^*$, which are obtained by reducing the domain of $\mathring{H}^*$ to only those functions $\Psi$ satisfying, for some or all couples of variables $x_i,x_j$, the condition that, if $y_{ij}:=x_i-x_j$ and $(x_1,x_2,\dots,x_n)\mapsto (y_{ij},y_2,\dots,y_N)$ is a regular change of variables, then
 \begin{equation}\label{eq:TMS-generic}
 \!\!\int_{\substack{p_{ij}\in\mathbb{R}^d \\ |p_{ij}|\leqslant R}} \;\widehat{\Psi}(p_{ij},p_2,\dots,p_N)\,\ud p_{ij}\;=\;(R-\frac{1}{a_{ij}})\,\xi_{ij}(p_2,\dots,p_N)+o(1)\,,\;\; R\to +\infty\,.
 \end{equation}
 In \eqref{eq:TMS-generic} $(p_{ij},p_2,\dots,p_N)$ are the conjugate Fourier variables to $(y_{ij},y_2,\dots,y_N)$,  the function $\xi_{ij}$ depends on $\Psi$ and on the considered couple $i,j$, but not on $R$, and the constant $a_{ij}\in\mathbb{R}\cup\{\infty\}$ is prescribed and is $R$-independent too. 
 Equation \eqref{eq:TMS-generic} above is one version of the so-called Ter-Martirosyan--Skornyakov condition. It is based on stringent physical heuristics (see the discussion in Section \ref{sec:history_TMS} and \eqref{eq:BP-generic}-\eqref{eq:BP-generic-xi} below) that allow one to interpret it as an interaction in the $(i,j)$-channel with \emph{zero range} and $s$-wave scattering length equal to $-a_{ij}^{-1}$.
 \item[3.] The TMS condition selects a ``physical'' extension $H_\alpha$ of $\mathring{H}$ labelled by the conventional parametrisation $\alpha\equiv(-a_{ij}^{-1})_{ij}$. The mathematical problem is then to recognise $H_\alpha$ as a self-adjoint extension of $\mathring{H}$, or to identify and classify its self-adjoint extensions, and then to investigate its spectral and stability properties.
 \end{itemize}

The scheme above, which was made explicit for the first time by Minlos \cite{Minlos-1987} in 1987 (based on an old seminal but very concise work of Minlos and Faddeev \cite{Minlos-Faddeev-1961-1} in 1961), has been since then the object of numerous investigations that we will quote in Section \ref{sec:history_TMS} The general question of the self-adjoint realisation of $H_\alpha$ and of its stability and spectral properties is still open. An amount of partial information is available for special cases of $N$-body systems.

In this work, beyond placing the problem into a historical perspective (see Section \ref{sec:history_TMS}), we discuss how the TMS condition \eqref{eq:TMS-generic} can be proved to be a \emph{self-adjointness condition}, based on the extension theory specifically tailored for semi-bounded symmetric operators, as developed by Kre{\u\i}n, Vi\v{s}ik, and Birman. (We review the main results of this theory in Appendix \ref{app:KVB} and we refer to \cite{M-KVB2015} for a comprehensive discussion.)

The use of the Kre{\u\i}n-Vi\v{s}ik-Birman theory is in fact at the basis of all recent studies on the self-adjoint realisation of the operator $H_\alpha$ selected by the TMS condition \eqref{eq:TMS-generic}: in this context, it allows one to reduce the problem of the self-adjointness of $H_\alpha$ on the Hilbert space $L^2(\mathbb{R}^{Nd})$ (with possible symmetries) to the problem of the self-adjointness of a suitable integral operator on the space of the functions $\xi_{ij}$'s appearing in the asymptotics \eqref{eq:TMS-generic}. What we do here is on the first place to elaborate in full detail the precise application of the Kre{\u\i}n-Vi\v{s}ik-Birman theory to the TMS condition, with a discussion and through intermediate results that to our knowledge are not present in the literature. Further, we put for the first time the emphasis on the crucial difference between the TMS condition as a \emph{point-wise} identity, and the same condition interpreted as a suitable \emph{functional} identity, as we now explain.

In Section \ref{sec:2body-point}, which contains the first group of our results, we discuss the TMS condition for the simplest composite system possible, consisting of two particles with point interaction. This is a well-studied quantum system that is completely understood within the standard self-adjoint extension theory a la von Neumann. First we re-obtain the well-known Hamiltonian of the system solely by means of the Kre{\u\i}n-Vi\v{s}ik-Birman theory, a procedure that we did not manage to find elsewhere in the literature. This yields an alternative equivalent characterisation of the whole class of self-adjoint extensions of the ``away-from-hyperplanes'' free Hamiltonian $\mathring{H}$ that represented the starting point of the analysis -- see step 1 of the general scheme above. We then show that imposing the TMS condition to functions in the domain of $\mathring{H}^*$ reproduces, for all possible values of the scattering length $a$, \emph{all} the self-adjoint extensions of $\mathring{H}$. In particular this shows that the TMS condition is a self-adjointness condition.

In Section \ref{sec:2+1}, where the second group of our results is presented, we follow the same approach for a more complicated system consisting of three particles with point interaction -- we develop our explicit discussion for the so-called ``2+1''-fermionic system, two identical fermions coupled with a third particle of different nature. Exploiting again the general results of the Kre{\u\i}n-Vi\v{s}ik-Birman theory, we find that the TMS condition is a self-adjointness condition, which selects a sub-class of extensions of $\mathring{H}^*$, only if it is given as a suitable functional identity. As a generic point-wise identity, instead, the TMS condition is in general not even a condition for a symmetric extension of $\mathring{H}$. Furthermore, we show that the issue of the self-adjoint realisation of the TMS condition does indeed boil down to the self-adjointness problem of a simpler integral operator (that acts on the so-called ``space of charges'') as normally given for granted in the literature, \emph{but} on different (more regular) functional spaces than those considered so far.

This brings us to the natural follow-up of the present analysis, which was in fact our original motivation and that we intend now to develop in a future work. There are indeed not completely understood discrepancies in the literature between the ranges of the particle masses in which the TMS condition for three-body systems is shown to be a self-adjointness condition, or instead to give rise to a symmetric extension of $\mathring{H}$ with its own family of self-adjoint extensions. Such discrepancies emerge between the operator-theoretic approach sketched above and an alternative approach through quadratic forms (i.e., the construction of a closed and semi-bounded quadratic form such that the function in the domain of self-adjoint operator that realises it display the TMS asymptotics \eqref{eq:TMS-generic}). Re-visiting the operator-theoretic approach in view of our present findings is likely to account for an explanation. We briefly elaborate on this point in the final Section \ref{sec:developments}



\section{A retrospective on Ter-Martirosyan--Skornyakov Hamiltonians for point interactions}\label{sec:history_TMS}

In this Section we present the historical emergence of TMS operators in the physics and mathematics  of point interactions. It therefore should \emph{not} be regarded as a complete review on the history of point interactions!

In the 1930's Quantum Mechanics began to be applied to the newly observed nuclear phenomena. At first, the decrease by a factor $10^{-5}$ from the atomic to the nuclear scale made it plausible to model the interaction among nucleons as a delta-like interaction. 

In 1932 Wigner \cite{Wigner-1933} calculated that the nuclear forces interaction must be of very short range and very strong magnitude. This led three years later first Bethe and Peierls \cite{Bethe_Peierls-1935,Bethe_Peierls-1935-np} and then Thomas \cite{Thomas1935} to describe the neutron-proton scattering by means of the two-body Schr\"{o}dinger equation in the approximation of a potential of very short range, an approach subsequently developed by Fermi \cite{Fermi-1936} and Breit  \cite{Breit-1947} with the introduction of the so-called ``delta pseudo-potential''. Formulated in modern terms, the celebrated ``\emph{Bethe-Peierls contact condition}'', which is still today ubiquitous in many formal physical treatments, prescribes on the basis of physical heuristics that the wave-function $\Psi(x_1,\dots,x_N)$ of $N$ three-dimensional particles subject to a two-body zero-range interaction of scattering length $a_{ij}$ among particles $i$ and $j$ behaves asymptotically as 
\begin{equation}\label{eq:BP-generic}
\Psi(x_1,\dots,x_N)\;\approx\;\Big(\frac{1}{|x_i-x_j|}-\frac{1}{a_{ij}}\Big)\qquad\textrm{ as }|x_i-x_j|\to 0
\end{equation}
where the point-wise limit \eqref{eq:BP-generic} is meant as
\begin{equation}\label{eq:BP-generic-xi}
\Psi(x_1,\dots,x_N)\;=\;\Big(\frac{1}{|x_i-x_j|}-\frac{1}{a_{ij}}\Big)\,\xi(Z)+o(1)
\end{equation}
for some function $\xi$ of the variable $Z$ in the hyperplane $\{x_i=x_j\}$. Clearly, what makes this approximation appealing, and computationally advantageous, is its dependence on few parameters only (the $a_{ij}$'s), instead of the complete knowledge of the interaction.

While Bethe and Peierls had studied the problem of \emph{two} low-energy nucleons with contact interaction and obtained \eqref{eq:BP-generic} for the two-body problem, Thomas had considered the \emph{three-body} problem showing that as the range of the two-body forces tends to zero the ground state of the three-body system can approach $-\infty$, even when the ground state energies of all two-body subsystems remain constant. This effect, referred to since then as the ``\emph{Thomas effect}'', was the first evidence that the deceptively simple three-body problem with zero-range interaction has a much richer (and potentially much more complicated) phenomenology than the analogous two-body problem.

The next extensive study of a system of \emph{three} low-energy nucleons appeared some 20 years later, in 1955, due to Ter-Martirosyan and Skornyakov \cite{TMS-1956}, two nuclear physicists who credited Landau for the ideas they exploited. They assumed that the Bethe-Peierls condition remains valid in \emph{each} two-body channel and they used it as boundary condition for solving the eigenvalue problem for the three-body Schr\"{o}dinger equation with formal delta-like two-body potentials. For brevity let us revisit their conclusion in the simplest case of spinless identical particles, for which they made use of standard centre-of-mass Jacobi coordinates $y_1=x_1-\frac{1}{2}(x_2+x_3)$, $y_{23}=x_2-x_3$ and introduced a formal delta-like potential $\delta(y_1)\delta(y_{23})$ for wave-functions $\Psi(y_1,y_{23})$. By suitably expressing $\widehat{\Psi}(p,q)$ in terms of an auxiliary function $\widehat{\xi}(p)$ (here $(p,q)$ and $(y_1,y_{23})$ are Fourier conjugate variables), they found that (in units $\hbar=1$ and particle mass $=1$) $\Psi$ is a bound state of energy $-E<0$ whenever $(\widehat{\xi},E)$ is a solution to the integral equation
\begin{equation}\label{eq:TMS-eqn-3bosons}
\alpha+2\pi^2\sqrt{\,\frac{3}{4}p^2+E\,}\,\widehat{\xi}(p)+2\int_{\mathbb{R}^3}\frac{\widehat{\xi}(q)}{p^2+q^2+p\cdot q+E}\,\ud q\;=\;0\,,
\end{equation}
where $\alpha:=-1/a$  and $a$ is the $s$-wave scattering length in each two-body channel.

Analogues to equation \eqref{eq:TMS-eqn-3bosons} were later found for other systems with point interactions, among which, to mention those that have received the largest attention, three-body systems with different symmetries (three distinguishable particles, two identical fermions plus a third particle of different nature, etc.), four-body systems with two distinct couples of identical fermions, and more generally $N+M$ systems with $N$ identical fermions of one type  plus $M$ identical fermions of another type. Each equation of this class is today referred to as a ``\emph{Ter-Martirosyan--Skornyakov equation}''.

The subject proved soon to be worth a deeper understanding, despite the effectiveness of the description. On the one hand the manipulations of Ter-Martirosyan and Skornyakov (as well as the previous ones by Bethe, Peierls, Thomas, Fermi, and Breit) were rather formal for mathematical standards and called for a more rigorous justification. On the other hand, an evidence of some sort of indeterminate physical description emerged when Danilov \cite{Danilov-1961} in 1961 observed that the TMS equation \eqref{eq:TMS-eqn-3bosons} has a solution $\widehat{\xi}$ for arbitrary values of $E$, with large momentum asymptotics
\begin{equation}\label{eq:Danilov-many}
\widehat{\xi}(p)\;=\;\frac{1}{p^2}\,\big(A_E\sin(s_0 \ln|p|)+B_E\cos(s_0\ln|p|)\big)+o\Big(\frac{1}{p^{2}}\Big)\,,
\end{equation}
where $s_0>0$ is an explicit universal constant and $A_E,B_E>0$ are two further constants that depend on $E$.
Inspired by ideas (mainly of Gribov) by which some additional ``experimental'' parameter that cannot be computed  using only two-body experimental data was needed for the full description of the three-body system, Danilov  proposed an ad hoc removal of this non-physical continuum of eigenvalues on $(-\infty,0)$ by constraining the solutions to \eqref{eq:TMS-eqn-3bosons} to have the form \eqref{eq:Danilov-many} with
\begin{equation}\label{eq:Danilov-restriction}
A_E\;=\;\beta B_E
\end{equation}
for some additionally prescribed parameter $\beta\in\mathbb{R}$. $\beta$ was in some vague sense given the meaning of a three-body parameter, as opposite to $\alpha$ in \eqref{eq:TMS-eqn-3bosons} which  is a parameter of the two-body problem for each couple of bosons. Under the restriction \eqref{eq:Danilov-restriction}, equation \eqref{eq:TMS-eqn-3bosons} has a discrete and infinite set of solutions $(\widehat{\xi}_n,E_n)$, with energies $-E_n\to -\infty$ as $n\to+\infty$ according to the asymptotics
\begin{equation}\label{eq:3bosons_EV_asymptotics_Thomas}
-E_n\;=-3 \exp\Big(\frac{2\pi n}{s_0}-\frac{2}{s_0}\arctan\frac{1}{\beta}\Big)(1+o(1))\,,
\end{equation}
a quantitative manifestation of the Thomas effect.

In modern mathematical terms, the phenomenon noted by Danilov is understood as follows: the three-body point-interaction Hamiltonian implicitly identified by Ter-Martirosyan and Skornyakov by means of the condition \eqref{eq:TMS-eqn-3bosons} for its eigenstates at given two-body scattering length $-\alpha^{-1}$, is \emph{not} a self-adjoint operator and it admits a one-parameter family of self-adjoint extensions, labelled by $\beta\in\mathbb{R}$; for each $\beta$, the corresponding self-adjoint Hamiltonian has a countable discrete spectrum accumulating exponentially to $-\infty$ with corresponding eigenfunctions that collapse onto the barycentre; the union of the negative spectra of all such self-adjoint extensions is the whole negative real line.

Motivated by the scheme of Ter-Martirosyan and Skornyakov for the three-body problem with point interaction and by Danilov's observation, Minlos and Faddeev \cite{Minlos-Faddeev-1961-1,Minlos-Faddeev-1961-2} in the same year 1961 provided essentially the whole explanation above, including the asymptotics \eqref{eq:Danilov-many} and \eqref{eq:3bosons_EV_asymptotics_Thomas}, in the form of two beautiful short announcements, albeit with no proofs or further elaborations. Theirs can be considered as the beginning of the mathematics of quantum systems with zero-range interactions. This is even more so because for the first time the problem was placed within a general mathematical framework, the theory of self-adjoint extensions of semi-bounded symmetric operator, that 
Kre{\u\i}n, Vi\v{s}ik, and Birman had developed between the mid 1940's and the mid 1950's (see Appendix \ref{app:KVB}).

A somewhat different approach characterised the start of the mathematical study of the \emph{two-body} problem. In 1960-1961, a few months before the works of  Minlos and Faddeev on the three-body problem, Berezin and Faddeev \cite{Berezin-Faddeev-1961} published the first rigorous analysis of a three-dimensional model with two particles coupled by a delta-like interaction. The emphasis was put in realising the formal Hamiltonian $-\Delta+\delta(x)$ as a self-adjoint extension of the  restriction $-\Delta|_{C^\infty_0(\mathbb{R}^3\setminus\{0\})}$ (in the relative variable $x=x_1-x_2$ between the two particles). Working in Fourier transform, they recognised that the latter operator  has deficiency indices $(1,1)$ and they characterised the whole family $\{H_\alpha\,|\,\alpha\in\mathbb{R}\}$ of its self-adjoint extensions as the operators 
\begin{equation}\label{eq:Berezin-Faddeev-1}
\widehat{(H_\alpha\psi)}(p)\;=\;p^2\widehat{\psi}(p)-\lim_{R\to\infty}\frac{1}{4\pi  R}\int_{\substack{\,p\in\mathbb{R}^3 \\ \! |p|<R}}\widehat{\psi}(q)\,\ud q
\end{equation}
defined on the domain of $L^2(\mathbb{R}^3)$-functions $\psi$ such that, as $R\to\infty$,
\begin{equation}\label{eq:Berezin-Faddeev-2}
\int_{\substack{\,p\in\mathbb{R}^3 \\ \! |p|<R}}\widehat{\psi}(q)\,\ud q\;=\;c \,(R+2\pi^2\alpha)+o(1)\quad\textrm{ and }\quad \int_{\mathbb{R}^3}|H_\alpha\psi|^2\ud x<\infty\,.
\end{equation}
(For a more direct comparison -- see \eqref{eq:TMS_cond_asymptotics_1} in the following -- we have replaced here the parameter $\alpha$ of the notation of \cite{Berezin-Faddeev-1961} with $-(8\pi^3\alpha)^{-1}$.)

As \eqref{eq:Berezin-Faddeev-1}-\eqref{eq:Berezin-Faddeev-2} were only announced with no derivation, with a sole reference to the monograph \cite{Akhiezer-Glazman-1961-1993} of Akhiezer and Glazman on linear operators in Hilbert space, we are to understand that Berezin and Faddeev came to their conclusion by methods of von Neumann's self-adjoint extension theory, as presented in \cite[Chapter VII]{Akhiezer-Glazman-1961-1993}, combined with explicit calculations in Fourier transform. This leaves the question open on why they did not approach the extension problem within the same language of Kre{\u\i}n, Vi\v{s}ik, and Birman, as used by Minlos and Faddeev for the three-body case. In this language, as we work out in Section \ref{sec:2body-point}, \eqref{eq:Berezin-Faddeev-1}-\eqref{eq:Berezin-Faddeev-2} would have emerged as a very clean application of the general theory and, most importantly, the asymptotics in \eqref{eq:Berezin-Faddeev-2} would have arisen with a natural and intimate connection with the TMS equation \eqref{eq:TMS-eqn-3bosons}. 
Berezin and Faddeev rather focused on re-interpreting the action of the Hamiltonian $H_\alpha$ as a renormalised rank-one perturbation of the free Laplacian, re-writing \eqref{eq:Berezin-Faddeev-1} in position coordinates as
\begin{equation}
H_\alpha\psi\;=\;-\Delta\psi-\frac{1}{4\pi\alpha}\,\lim_{R\to\infty}\frac{1}{\,2\pi^2+R/\alpha\,}\frac{\sin R|x|}{|x|}\int_{\mathbb{R}^3}\frac{\sin R|y|}{|y|}\psi(y)\,\ud y\,.
\end{equation}
We conjecture that they did not know the old work of Bethe and Peierls for two nucleons, or they did not consider it relevant in their context, for no word is spent in \cite{Berezin-Faddeev-1961} to derive the singularity $\psi(x)\sim |x|^{-1}$ as $|x|\to 0$ from their asymptotics \eqref{eq:Berezin-Faddeev-2}.

With the subsequent theoretical and experimental advances in nuclear physics -- the initial playground for models of point interactions -- it became clear that the assumption of zero range was only a crude simplification of no fundamental level. 
The lack of a physically stringent character for the idealisation of zero range in experimentally observed quantum-mechanical systems, and the somewhat obscure emergence of the unboundedness from below for the self-adjoint realisations of the three-body Hamiltonian, decreased the physical interest towards point interactions and left their rigorous study in a relatively marginal position, and the approach of Ter-Martirosyan and Skornyakov quiescent. 
Moreover, after Faddeev published in 1963 his fundamental work \cite{Faddeev-1963-eng-1965-3body} on the three-body problem with regular two-body forces, the concern of the physicists switched over to the numerical solutions of the corresponding Faddeev equations.
In the Russian physical literature, mainly under the input of Faddeev, methods and models of point interactions, albeit not fully rigorous, moved their applicability  to atomic and molecular physics, a mainstream that ideally culminates with the late 1970's monograph of Demkov and Ostrovskii \cite{Demkov-Ostrovskii-book} on the ``zero-range potentials'' and their application to atomic physics. 
The use of formal delta-like potentials remained for some decades   as a tool for a formal first-order perturbation theory; in addition, the Kre{\u\i}n-Vi\v{s}ik-Birman self-adjoint extension theory  lost ground to von Neumann's theory in the literature in English language on the mathematics for quantum mechanics -- it rather evolved in more modern forms in application to boundary value problems for partial differential equations, mainly in the modern theory of boundary triplets.

It is the merit of Albeverio, Gesztesy, and H\o{}egh-Krohn, and their collaborators (among whom, Streit and Wu), in the end of the 1970's and throughout the 1980's, to have unified an amount of previous investigations by establishing a proper mathematical branch on rigorous models of point interactions, with a systematic study of \emph{two-body} Hamiltonians  and of \emph{one-body} Hamiltonians with finite or infinitely many \emph{fixed centers} of point interaction. We refer to the monograph \cite{albeverio-solvable} for a comprehensive overview on this production, and especially to the end-of-chapter notes in \cite{albeverio-solvable} for a detailed account of the previous contributions.
The main tools in this new mainstream were: von Neumann's extension theory on the first place (hence with no reference any longer to the methods of Kre{\u\i}n-Vi\v{s}ik-Birman), by which point interaction Hamiltonians were constructed as self-adjoint extensions of the restriction of the free Laplacian to functions that vanish in a neighbourhood of the point where the interaction is supported; resolvent identities (of Kre{\u\i}n and of  Konno-Kuroda type, see \cite[Appendices A and B]{albeverio-solvable}) by which these self-adjoint extensions were recognised to be finite-rank perturbations of the free Laplacian, in the resolvent sense, and were also re-obtained by  resolvent limits of Schr\"{o}dinger Hamiltonians with shrinking potentials; plus an amount of additional methods (Dirichlet quadratic forms, non-standard analysis methods, renormalisation methods) for specific problems.

Let us emphasize, in particular, that the original heuristic arguments of Bethe and Peierls and their two-body contact condition  find a rigorous ground based on the fact, which can be proved within von Neumann's extension theory (see, e.g., \cite[Theorems I.1.1.1 and I.1.1.3]{albeverio-solvable}), that any self-adjoint extension of $\Delta|_{C^\infty_0(\mathbb{R}^3\setminus\{0\})}$ on $L^2(\mathbb{R}^3)$ has a domain whose elements behave as $\psi(x)\sim(|x|^{-1}+\alpha)$ as $|x|\to 0$, as an $s$-wave (hence a ``low-energy'') boundary conditions, for some $\alpha\in(-\infty,+\infty]$.

As for the initial three-body problem with two-body point interaction, it finally re-gained centrality from the mathematical point of view (while physically a stringent experimental counterpart was still lacking) around the end of the 1980's and throughout the 1990's. This was  first due to Minlos and his school \cite{Minlos-1987,Minlos-Shermatov-1989,mogilner-shermatov-PLA-1990,Menlikov-Minlos-1991,Menlikov-Minlos-1991-bis,Minlos-TS-1994,Shermatov-2003} (among which Melnikov, Mogilner, and Shermatov), by means of the operator-theoretic approach used for three identical bosons by Minlos and Faddeev, and slightly later due to Dell'Antonio and his school \cite{Teta-1989,dft-Nparticles-delta,DFT-proc1995} (among which Figari and Teta), with an approach based on  quadratic forms, where the ``physical'' energy form  is first regularised by means of an ultra-violet cut-off and a suitable renormalisation procedure, and then is shown to be realised by a self-adjoint Hamiltonian. An alternative direction was started further later by Pavlov and a school that included Kuperin, Makarov, Melezhik, Merkuriev, and Motovilov, \cite{Kuperin-Makarov-Merk-Motovilov-Pavlov-1989-JMP1990,Makarov-Melezhik-Motovilov-1995}, by indroducing internal degrees of freedom, i.e., a spin-spin contact interaction, so as to realise semi-bounded below three-body Hamiltonians.

After a further period of relative quiescence, the subject has been experiencing a new boost, due to the last decade's rapid progress in the manipulation techniques for ultra-cold atoms and, in particular, for tuning the effective $s$-wave scattering length by means of a magnetically induced Feshbach resonance \cite[Section 5.4.2]{pethick02}.
This has made it possible, among others, to prepare and study ultra-cold gases in   the so-called ``\emph{unitary regime}'' \cite{Castin-Werner-2011_-_review}, i.e., the case of negligible two-body interaction range and huge, virtually infinite, two-body scattering length (both lengths being compared to a standard reference length such as the Bohr radius). 
In such a regime, unitary gases show properties, including superfluidity, that have the remarkable feature of being universal in several respects \cite{Braaten-Hammer-2006}, and are under active experimental and theoretical investigation. As we do not have space here for an  outlook on such an active field,  we refer to the overview given in the introductory sections of the works \cite{michelangeli-schmidbauer-2013,MP-2015-2p2} and to the references therein. Let us only underline that from the experimental point of view, zero-range interactions in ultra-cold atom physics are today  far from being just an idealisation of real-world two-body potentials with small support and in many realisations the zero-range, delta-like character of the interaction turns out to be an extremely realistic and in fact an unavoidable description.

In turn, all this has brought new impulse and motivations to the already developing  mathematical research on the subject, with a series of fundamental contributions in the last few years  \cite{Minlos-2011-preprint_May_2010,Minlos-2010-bis,Minlos-2012-preprint_30sett2011,Finco-Teta-2012,CDFMT-2012,Minlos-2012-preprint_1nov2012,Minlos-RusMathSurv-2014,CDFMT-2015}, many of which provide rigorous ground to experimental or numerical evidence on the physical side.

\section{Two-body point interaction \`a la Ter-Martirosyan--Skornyakov}\label{sec:2body-point}

The Hamiltonian of point interaction between two particles in three dimension is well known since the first rigorous attempt \cite{Berezin-Faddeev-1961} by Berezin and Faddeev  in 1961, which we have already mentioned in Section \ref{sec:history_TMS}, and the seminal work \cite{AHK-1981-JOPTH} by Albeverio and H\o{}egh-Krohn  in 1981. In 
\cite[Chapter I.1]{albeverio-solvable} one can find the complete discussion of the self-adjoint realisation of this operator, its explicit domain and action,  its resolvent, its spectral properties, its approximation by short-range potentials, and its scattering theory.

In the first part of this Section we shall re-obtain this Hamiltonian and its main properties within the self-adjoint extension scheme of 
Kre{\u\i}n-Vi\v{s}ik-Birman, as opposite to von Neumann's scheme used in the above works. We follow this line both for general reference, because to our knowledge this approach has never been  worked out in the literature, and above all because we need to establish the grounds for the second part of this Section, where we shall realise the point interaction \`a la Ter-Martirosyan--Skornyakov. For the tools from the Kre{\u\i}n-Vi\v{s}ik-Birman theory we shall make use of, we refer to the Appendix \ref{app:KVB} and, more diffusely, to the work \cite{M-KVB2015}.

\subsection{Point interaction Hamiltonian through the Kre{\u\i}n-Vi\v{s}ik-Birman theory}\label{subsec:2delta}

The starting point is the  operator
\begin{equation}\label{eq:def_Hdot_start}
\mathring{H}\;=\;-\Delta\,,\qquad\mathcal{D}(\mathring{H})\;=\;H^2_0(\mathbb{R}^3\!\setminus\!\{0\})\,,
\end{equation}
which is clearly a densely defined, symmetric, closed, and positive operator on the Hilbert space $L^2(\mathbb{R}^3)$. The variable $x\in\mathbb{R}^3$ has the meaning of relative variable between the two particles: after removing the centre of mass of the two-body system, the only relevant problem is in the relative variable. $\mathring{H}$ is the closure of the negative Laplacian restricted to the smooth functions compactly supported  away from the origin, and
\begin{equation}
H^2_0(\mathbb{R}^3\!\setminus\!\{0\})\;=\;\overline{\,C^\infty_0(\mathbb{R}^3\!\setminus\!\{0\})\,}^{\|\,\|_{H^2}}.
\end{equation}
The space above is clearly a closed subspace of $H^2(\mathbb{R}^2)$, and it is also proper, as shown in \eqref{eq:actionHdot_to_f} and \eqref{eq:H20approx} below.

The free Hamiltonian on $C^\infty_0(\mathbb{R}^3\!\setminus\!\{0\})$ is the natural starting point when one aims at constructing a singular interaction supported only at $x=0$, and any self-adjoint extension of this operator has the natural interpretation of a ``candidate'' Hamiltonian for the point interaction.

Throughout this discussion it will be convenient to work in Fourier transform. We therefore re-write \eqref{eq:def_Hdot_start} by means of the following simple Lemma:

\begin{lemma}{Lemma}\label{lemma:dom_Hdot_Ftransform}
For the operator $\mathring{H}$ defined in \eqref{eq:def_Hdot_start} one has
\begin{itemize}
 \item[(i)]
 the domain and the action of $\mathring{H}$ are given by
\begin{equation}
\mathcal{D}(\mathring{H})\;=\;\Big\{f\in H^2(\mathbb{R}^3)\,,\int_{\mathbb{R}^3}\widehat{f}(p)\,\ud p=0\Big\}\,,\qquad \widehat{(\mathring{H}f)}(p)\;=\;p^2\widehat{f}(p)\,;  \label{eq:actionHdot_to_f}
\end{equation}
\item[(ii)] the Friedrichs extension of $\mathring{H}$ is given by
\begin{equation}\label{eq:HringF}
\mathcal{D}(\mathring{H}_F)\;=H^2(\mathbb{R}^3)\;, \qquad \widehat{(\mathring{H}_Ff)}(p)\;=\;p^2\widehat{f}(p)\,.
\end{equation}
\end{itemize}
\end{lemma}

\begin{proof}
By suitable approximation arguments (see Appendix \ref{app:approx}), we have
\begin{equation}\label{eq:H20approx}
\mathcal{D}(\mathring{H})\;\equiv\;{H^2_0(\mathbb{R}^3\!\setminus\!\{0\})}\;=\;\big\{ f\in H^2(\mathbb{R}^3)\,\big|\, f(0)=0\big\}\,.
\end{equation}
Moreover, since $(1+p^2)\widehat f$ and $(1+p^2)^{-1}$ are in $L^2(\mathbb{R}^3)$, then  $\widehat{f}\in L^1(\mathbb{R}^3)$. Hence
\begin{equation}
0\;=\;f(0)\;=\;\int_{\mathbb{R}^3}{\widehat f(p) \ud p}
\end{equation}
which, together with $\widehat{(-\Delta f)}=p^2\widehat{f}\in H^2(\mathbb{R}^3)$, proves part (i). As for part (ii), we first observe that the form domain of $\mathring{H}$, which is the completion of $\mathcal{D}(\mathring{H})$ in the $H^1$-norm, is precisely
\begin{equation}
\mathcal{D}[\mathring{H}]\;=\overline{\,H^2_0(\mathbb{R}^3\!\setminus\!\{0\})\;}^{\,\|\,\|_{H^2}}=\;\;H^1_0({\mathbb{R}}^3\!\setminus\!\{0\})\,.
\end{equation}
Since 
\begin{equation}\label{useful_inclusion}
H^2(\mathbb{R}^3)\;\subset\; H^1_0(\mathbb{R}^3\!\setminus\!\{0\})\;=\;H^1(\mathbb{R}^3)
\end{equation}
(see Appendix \ref{app:inclusions}) and 
$H^2(\mathbb{R}^3)$ is the domain of a self-adjoint extension of $\mathring{H}$, namely the self-adjoint $-\Delta$  on $\mathbb{R}^3$, we conclude that $-\Delta$ must be the Friedrichs extension $\mathring{H}_F$ of $\mathring{H}$, owing to the characterisation of  $\mathring{H}_F$ as the unique self-adjoint extension of $\mathring{H}$ whose operator domain is contained in $\mathcal{D}[\mathring{H}]$.
\end{proof}

As every semi-bounded and densely defined symmetric operator, $\mathring{H}$ admits self-adjoint extensions.
We are after the family of such extensions.
The first step is to determine the adjoint of $\mathring{H}$.

\begin{theorem}{Proposition}\label{prop:D_Hdostar}
Let $\lambda>0$.
\begin{itemize}
 \item[(i)] One has
 \begin{equation}\label{eq:ker_hring*}
 \begin{split}
 \ker (\mathring{H}^*+\lambda\mathbbm{1})\;&=\;\Big\{ u_\xi\in L^2(\mathbb{R}^3)\textrm{ of the form }\,\widehat{u}_\xi(p)=\frac{\xi}{p^2+\lambda}\,\Big|\,\xi\in\mathbb{C}\Big\} \\
 &=\;\mathrm{span}\,\big\{\big(\,p^2+\lambda)^{-1}\big){\textrm{\LARGE$\check{\,}$\normalsize}}\, \big\}
 \end{split}
 \end{equation}
 \item[(ii)] The domain and the action of the adjoint of $\mathring{H}$ are given by
 \begin{eqnarray}
  \mathcal{D}(\mathring{H}^*)&\!=\!&\left\{g\in L^2(\mathbb{R}^3)\left|\!
  \begin{array}{c}
  \widehat{g}(p)=\displaystyle\widehat{f}(p)+\frac{\eta}{(p^2+\lambda)^2}+\frac{\xi}{p^2+\lambda} \\
  f\in\mathcal{D}(\mathring{H})\,,\quad \eta,\xi\in\mathbb{C}
  \end{array}\!\!\!\right.\right\} \label{eq:decompositionD_Hdostar}\\
  (\widehat{(\mathring{H}^*+\lambda) g)}\,(p)&\!=\!& (p^2+\lambda) \,\Big(\widehat{f}(p)+\frac{\eta}{(p^2+\lambda)^2}\Big) \label{eq:actionHdotstar_to_g-f}\\
  \widehat{(\mathring{H}^* g)}(p)&\!=\!& p^2 \widehat{g}(p)-\xi\,. \label{eq:actionHdotstar_to_g}
 \end{eqnarray}
\end{itemize}
\end{theorem}

\begin{remark}{Remark}
 The decomposition \eqref{eq:decompositionD_Hdostar} of the generic element $g\in\mathcal{D}(\mathring{H}^*)$ depends on the chosen $\lambda$, but of course $\mathcal{D}(\mathring{H}^*)$ does not, nor does the outcome of $\mathring{H}^*$ applied to $g$, as one sees from  \eqref{eq:actionHdotstar_to_g}.
\end{remark}

\begin{proof}[Proof of Proposition \ref{prop:D_Hdostar}]
In order to apply the general decomposition formulas of Lemma \ref{lemma:krein_decomp_formula} we need to deal with an operator with positive bottom. To this aim we introduce the auxiliary operator
\begin{equation}
\mathring{H}_{\lambda}\;:=\;\mathring{H}+\lambda\mathbbm{1}
\end{equation}
which is by construction densely defined, symmetric, and closed, and with bottom $m(S)=\lambda$. Clearly, $\mathcal{D}(\mathring{H}_{\lambda})=\mathcal{D}(\mathring{H})$ and $\mathring{H}_{\lambda}^*=\mathring{H}^*+\lambda\mathbbm{1}$. Since $\ker (\mathring{H}^*_{\lambda})\;=\ran(\mathring{H}_{\lambda})^\perp$, then $u\in\ker (\mathring{H}^*_{\lambda})\;$ if and only if
\begin{equation*}
0\;=\;\int_{\mathbb{R}^3}{(\mathring{H}_{\lambda}f) \,u\, \ud x}\;=\;\int_{\mathbb{R}^3}{\widehat{(\mathring{H}_{\lambda}f)}\,\widehat u \,\ud p}=\int_{\mathbb{R}^3}{\widehat f\, (p^2+\lambda)\,\widehat{u} \,\ud p}\qquad\forall f\in\mathcal{D}(\mathring{H})\,,
\end{equation*}
which by \eqref{eq:actionHdot_to_f} and a standard localisation argument yields \eqref{eq:ker_hring*}.
Because of $(\mathring{H}_\lambda)_F=\mathring{H}_F+\lambda\mathbbm{1}$ and \eqref{eq:HringF}, we have that
\begin{equation}\label{eq:HringF-1}
\widehat{(\mathring{H}_\lambda)_F^{-1}u}\;=\;(p^2+\lambda)^{-1}\widehat{u}\,.
\end{equation}
This, together with the decomposition formula \eqref{eq:DomS*} discussed in Appendix \ref{app:KVB} (Lemma \ref{lemma:krein_decomp_formula}), and the characterisation  \eqref{eq:ker_hring*} of $\ker \mathring{H}_\lambda^*$ yield immediately \eqref{eq:decompositionD_Hdostar}.
The decomposition \eqref{eq:DomS*} also implies that the action of $\mathring{H}_\lambda^*$ on a generic element $f+(\mathring{H}_\lambda)_F^{-1}u_\eta+u_\xi\in\mathcal{D}(\mathring{H}_\lambda^*)$ is the same as the action of $(\mathring{H}_\lambda)_F$ on the component $f+(\mathring{H}_\lambda)_F^{-1}u_\eta\in\mathcal{D}((\mathring{H}_\lambda)_F)$, while $\mathring{H}_\lambda^* u_\xi=0$: this is precisely  \eqref{eq:actionHdotstar_to_g-f}. As for \eqref{eq:actionHdotstar_to_g}, it follows directly from  \eqref{eq:decompositionD_Hdostar} and \eqref{eq:actionHdotstar_to_g-f}.
\end{proof}

\begin{remark}{Remark}
The decomposition formula \eqref{eq:DomSF}, together with \eqref{eq:HringF-1} above, gives
\begin{equation}\label{eq:Hring_lambda_F}
\mathcal{D}((\mathring{H}_\lambda)_F)\;=\;\left\{g\in L^2(\mathbb{R}^3)\left|\!
  \begin{array}{c}
  \widehat{g}(p)=\displaystyle\widehat{f}(p)+(p^2+\lambda)^{-2}\eta \\
  f\in\mathcal{D}(\mathring{H})\,,\quad \eta\in\mathbb{C}
  \end{array}\!\!\!\right.\right\} 
\end{equation}
Therefore, expression \eqref{eq:decompositionD_Hdostar} shows that a generic $g\in\mathcal{D}(\mathring{H}^*)$ 
is in general less regular than $H^2(\mathbb{R}^3)$, for only the component $\mathcal{F}^{-1}(\widehat{f}+(p^2+\lambda)^{-2}\eta)=f+(\mathring{H}_\lambda)_F^{-1}u_\eta$ is in $\mathcal{D}((\mathring{H}_\lambda)_F)=\mathcal{D}(\mathring{H}_F)=H^2(\mathbb{R}^3)$, whereas the component $u_\xi$ is not. Related to that, \eqref{eq:actionHdotstar_to_g} shows that whereas $\mathring{H}^*g\in L^2(\mathbb{R}^3)$, none of the two \emph{distributions} $-\Delta g$ and $(2\pi)^{3/2}\xi\delta(x)$ whose difference gives precisely $\mathring{H}^*g$ is realised as a square-integrable function (thus, in the difference the two non-square-integrable singularities cancel out).
\end{remark}

With the above knowledge of  $\mathring{H}^*+\lambda\mathbbm{1}$  and  $\ker (\mathring{H}^*+\lambda\mathbbm{1})$ the Kre{\u\i}n-Vi\v{s}ik-Birman theory provides an explicit prescription to restrict $\mathring{H}_\lambda^*$ so as to find the whole family of self-adjoint extensions of $\mathring{H}_\lambda$, and hence of $\mathring{H}$.  

\begin{theorem}{Theorem}\label{thm:1p1_extensions_Birman}~
\begin{itemize}
 \item[(i)] The self-adjoint extensions of the operator $\mathring{H}$ on $L^2(\mathbb{R}^3)$ constitute the one-parameter family $\{\mathring{H}^{(\tau)}\,|\,\tau\in\mathbb{R}\cup\{\infty\}\}$ where $\mathring{H}^{(\infty)}$ is  the Friedrichs extension $\mathring{H}_F$, that is, 
 \begin{equation}\label{eq:Hring-inf}
\mathcal{D}(\mathring{H}^{(\infty)})\;=H^2(\mathbb{R}^3)\;, \qquad \widehat{(\mathring{H}^{(\infty)} f)}(p)\;=\;p^2\widehat{f}(p)\,,
\end{equation}
whereas, for $\tau\in\mathbb{R}$,
 \begin{eqnarray}
   \mathcal{D}(\mathring{H}^{(\tau)})\!&=&\!\!\left\{g\in L^2(\mathbb{R}^3)\!\left|\!
  \begin{array}{c}
  \widehat{g}(p)=\displaystyle\widehat{f}(p)+\frac{\tau\,\xi}{(p^2+\lambda)^2}+\frac{\xi}{p^2+\lambda} \\
  \xi\in\mathbb{C}\,,\quad f\in\mathcal{D}(\mathring{H}) 
  \end{array}\!\!\!\!\right.\right\} \label{eq:D_HringTau}\\
  ((\mathring{H}^{(\tau)}+\lambda\mathbbm{1}) \,g)^{\widehat{\;}}\,(p)\!\!&=&\!\! (p^2+\lambda) \,\Big(\widehat{f}(p)+\tau\,\frac{\xi}{(p^2+\lambda)^2}\Big) \label{eq:Hring-tau_action1}\\
  \widehat{(\mathring{H}^{(\tau)} g)}(p)\!\!&=&\!\! p^2 \widehat{g}(p)-\xi\,, \label{eq:Hring-tau_action2}
 \end{eqnarray}
 where $\lambda>0$ is arbitrary.
 \item[(ii)] Each extension $\mathring{H}^{(\tau)}$ is semi-bounded below. In particular, for the bottom $m(\mathring{H}^{(\tau)})$ of $\mathring{H}^{(\tau)}$ one has
  \begin{equation}\label{eq:positiveHring-tau_iff_positve_tau}
 \begin{split}
 m(\mathring{H}^{(\tau)})\;\geqslant \;0\quad&\Leftrightarrow\quad \tau\;\geqslant\; 0 \\
 m(\mathring{H}^{(\tau)})\;> \;0\quad&\Leftrightarrow\quad \tau\;>\; 0\,.
 \end{split}
 \end{equation}
 \item[(iii)] For  each $\tau\in\mathbb{R}$ the quadratic form of the extension $\mathring{H}^{(\tau)}$ is given by
\begin{eqnarray}
   \mathcal{D}[\mathring{H}^{(\tau)}]\!&=&\!\! H^1(\mathbb{R}^3)\dotplus \ker(\mathring{H}^*+\lambda\mathbbm{1})  \label{eq:Hring-tau_form1}\\
  \mathring{H}^{(\tau)}[\phi+u_\xi]\!\!&=&\!\! \|\nabla\phi\|_2^2-\lambda\|\phi+u_\xi\|_2^2+\lambda\|\phi\|_2^2+\tau\frac{\;\pi^2}{\sqrt{\lambda}}|\xi|^2 \label{eq:Hring-tau_form2}
 \end{eqnarray}
 for any $\phi\in H^1(\mathbb{R}^3)$ and any $\widehat{u}_\xi(p)=(p^2+\lambda)^{-1}\xi$, $\xi\in\mathbb{C}$, 
where $\lambda>0$ is arbitrary.
\end{itemize}
\end{theorem}

\begin{proof}
Fixed $\lambda>0$, by Theorem \ref{thm:VB-representaton-theorem_Tversion} the self-adjoint extensions of $\mathring{H}+\lambda\mathbbm{1}$ are one-to-one with the self-adjoint operators on Hilbert subspaces of $\ker(\mathring{H}+\lambda\mathbbm{1})^*$, which is the one-dimensional space found in \eqref{eq:ker_hring*}. 
The generic case is that of a self-adjoint operator acting on the whole $\ker(\mathring{H}+\lambda\mathbbm{1})^*$, that is, the map $T_\tau: u_\xi\mapsto\tau u_\xi$ of multiplication by the scalar $\tau\in\mathbb{R}$. In this case the expression \eqref{eq:ST} for $\mathcal{D}(\mathring{H}^{(\tau)}+\lambda\mathbbm{1})$ ($=\mathcal{D}(\mathring{H}^{(\tau)})$) reads, by means of \eqref{eq:ker_hring*} and \eqref{eq:HringF-1}, precisely as \eqref{eq:D_HringTau}. Then 
\eqref{eq:Hring-tau_action1} follows from \eqref{eq:D_HringTau} and the fact that $\mathring{H}^{(\tau)}+\lambda\mathbbm{1}=(\mathring{H}^*+\lambda\mathbbm{1})\!\upharpoonright_{\mathcal{D}(\mathring{H}^{(\tau)})}$, and \eqref{eq:Hring-tau_action2} is an immediate consequence of \eqref{eq:Hring-tau_action1}. 
The case of the subspace $\{0\}$ of $\ker(\mathring{H}+\lambda\mathbbm{1})^*$ corresponds by Proposition \ref{prop:parametrisation_SF_SN_Tversion} to the Friedrichs extension $(\mathring{H}+\lambda\mathbbm{1})_F=\mathring{H}_F+\lambda\mathbbm{1}$, where $\mathring{H}_F$ has been determined in \eqref{eq:HringF} of Lemma \ref{lemma:dom_Hdot_Ftransform}. Thus,  re-writing the generic $g\in\mathcal{D}(\mathring{H}^{(\tau)})$ as $\widehat{g}=\widehat{f}+(p^2+\lambda)^{-2}\widetilde{\xi}+\tau^{-1}(p^2+\lambda)^{-1}\widetilde{\xi}$ for arbitrary $f\in\mathcal{D}(\mathring{H})$ and $\widetilde{\xi}\in\mathbb{C}$, and comparing it with \eqref{eq:Hring_lambda_F}, one recognises that  $\mathcal{D}(\mathring{H}^{(\tau)}+\lambda\mathbbm{1})=\mathcal{D}(\mathring{H}_F+\lambda\mathbbm{1})$ when $\tau=\infty$.
Therefore, $\mathring{H}^{(\infty)}=\mathring{H}_F$ and \eqref{eq:Hring-inf} follows from \eqref{eq:HringF}.
This concludes the proof of part (i). 
The semi-boundedness of each $\mathring{H}^{(\tau)}$ follows 
by Proposition \ref{cor:finite_deficiency_index}, or also Proposition \ref{cor:finite-dimensional}.
Conditions \eqref{eq:positiveHring-tau_iff_positve_tau} for $\mathring{H}^{(\tau)}$ follow by the general conditions \eqref{eq:positiveSBiffpositveB-1_Tversion}, using the fact that $m(\mathring{H}^{(\tau)}+\lambda\mathbbm{1})=m(\mathring{H}^{(\tau)})+\lambda$ and that $\lambda>0$ is arbitrary, thus also part (ii) is proved. Last, we observe that $\mathcal{D}[\mathring{H}^{(\tau)}]=\mathcal{D}[\mathring{H}^{(\tau)}+\lambda\mathbbm{1}]$ and $\mathring{H}^{(\tau)}[g]=(\mathring{H}^{(\tau)}+\lambda\mathbbm{1})[g]-\lambda\|g\|_2^2$, thus one deduces \eqref{eq:Hring-tau_form1}-\eqref{eq:Hring-tau_form2} from \eqref{eq:decomposition_of_form_domains_Tversion} of Theorem \ref{thm:semibdd_exts_form_formulation_Tversion} applied to $\mathring{H}^{(\tau)}+\lambda\mathbbm{1}$: formula \eqref{eq:Hring-tau_form1} is an immediate consequence of \eqref{eq:decomposition_of_form_domains_Tversion};
concerning \eqref{eq:Hring-tau_form2}, formula \eqref{eq:decomposition_of_form_domains_Tversion} prescribes the contributions $(\mathring{H}_F+\lambda\mathbbm{1})[\phi+u_\xi]$ and $-\lambda\|\phi+u_\xi\|_2^2$ to $\mathring{H}^{(\tau)}[\phi+u_\xi]$, which are the first three summands in the r.h.s.~of \eqref{eq:Hring-tau_form2}, plus the term
\[
T_\tau[u_\xi]\;=\;\tau\int_{\mathbb{R}^3}\Big|\frac{\xi}{p^2+\lambda}\Big|^2\,\ud p\;=\;\tau\,\frac{\;\pi^2}{\sqrt{\lambda}}\,|\xi|^2
\]
which is the fourth summand. This  completes the proof of part (iii).
\end{proof}

\begin{theorem}{Corollary}\label{cor:delta3D_xcoordinates}
For each self-adjoint extension $\mathring{H}^{(\tau)}$ and for arbitrary $\lambda>0$ one has
\begin{eqnarray}
\mathcal{D}(\mathring{H}^{(\tau)})\!\!\!&=&\!\!\!\Big\{g=\phi+\phi(0)\,\frac{8\pi\sqrt{\lambda}}{\tau}\,G_\lambda\,\Big|\,\phi\in H^2(\mathbb{R}^3)\,,\;G_\lambda(x)=\frac{e^{-\sqrt{\lambda}|x|}}{4\pi|x|}\Big\} \label{eq:D_HringTau-xversion}\\
(\mathring{H}^{(\tau)}+\lambda\mathbbm{1})\,g\!\!\!&=&\!\!\! (-\Delta+\lambda)\,\phi \label{eq:Hring-tau_action1-xversion}
\end{eqnarray}
and
\begin{eqnarray}
\mathcal{D}[\mathring{H}^{(\tau)}]\!\!\!&=&\!\!\!\Big\{g=\phi+\eta\,G_\lambda\,\Big|\,\phi\in H^1(\mathbb{R}^3)\,,\;\eta\in\mathbb{C}\,,\;G_\lambda(x)=\frac{e^{-\sqrt{\lambda}|x|}}{4\pi|x|}\Big\} \label{eq:Hring-tau_form1_xversion}\\
\mathring{H}^{(\tau)}[\phi+\eta\,G_\lambda]\!\!\!&=&\!\!\! -\lambda\|\phi+\eta\,G_\lambda\|_2^2+\|\nabla\phi\|_2^2+\lambda\|\phi\|_2^2+\frac{\tau}{8\pi\sqrt{\lambda}}\,|\eta|^2\,. \label{eq:Hring-tau_form2_xversion}
\end{eqnarray}
\end{theorem}

\begin{proof}
Since
\begin{equation}\label{eq:Ftransforms_a}
\Big(\frac{1}{(p^2+\lambda)^2}\Big){\textrm{\LARGE$\check{\,}$\normalsize}}(x)\;=\;\sqrt{\frac{\pi}{8\lambda}}\,e^{-\sqrt{\lambda}|x|}\quad\textrm{ and }\quad\Big(\frac{1}{p^2+\lambda}\Big){\textrm{\LARGE$\check{\,}$\normalsize}}(x)\;=\;(2\pi)^{3/2}\,\frac{\,e^{-\sqrt{\lambda}|x|}}{4\pi|x|}\,,
\end{equation}
the inverse Fourier transform of a generic $\widehat{g}$ of the form \eqref{eq:D_HringTau} reads
\begin{equation}\label{eq:g-inv-transf-1}
g=\phi+\xi\,(2\pi)^{3/2}\,G_\lambda
\end{equation}
where
\begin{equation}\label{eq:G-phi-x}
G_\lambda(x)\;:=\;\frac{e^{-\sqrt{\lambda}|x|}}{4\pi|x|}\,,\quad\widehat{\phi}\;:=\;\widehat{f}+\frac{\tau\,\xi}{(p^2+\lambda)^2}\,,\quad\textrm{ i.e., }\quad \phi\;=\;f+\tau\,\xi\,\sqrt{\frac{\pi}{8\lambda}}\,e^{-\sqrt{\lambda}|x|}\,.
\end{equation}
From \eqref{eq:G-phi-x} one computes
\begin{equation}\label{eq:phi0_xi}
\phi(0)\;=\;\tau\,\xi\,\sqrt{\frac{\pi}{8\lambda}}
\end{equation}
and plugging this into  \eqref{eq:g-inv-transf-1} one obtains  \eqref{eq:D_HringTau-xversion}.  \eqref{eq:Hring-tau_action1-xversion} is the inverse Fourier transform of \eqref{eq:Hring-tau_action1}.  \eqref{eq:Hring-tau_form1_xversion} follows by taking the inverse Fourier transform in \eqref{eq:Hring-tau_form1}, using the characterisation \eqref{eq:ker_hring*} for $\ker(\mathring{H}^*+\lambda\mathbbm{1})$ and formulas \eqref{eq:Ftransforms_a} and \eqref{eq:G-phi-x}. \eqref{eq:Hring-tau_form2_xversion} is a straightforward re-writing of \eqref{eq:Hring-tau_form2}.
\end{proof}

\begin{remark}{Remark}\label{rem:Htau_Halpha}
In Corollary \ref{cor:delta3D_xcoordinates} above we have re-derived, apart from an obvious re-definition of the extension parameter $\tau$, the well-known formulas for $\mathring{H}^{(\tau)}$ previously obtained in the literature by means of von Neumann's extension theory. Indeed, each $\mathring{H}^{(\tau)}$ is precisely the extension $H_\alpha$ discussed in \cite[Chapter I.1]{albeverio-solvable}, where
\begin{equation}\label{eq:alpha-tau}
\alpha\;=\;\frac{\tau-2\lambda}{\,8\pi\sqrt{\lambda}\,}\,;
\end{equation}
by means of \eqref{eq:alpha-tau}, the expression \eqref{eq:D_HringTau-xversion} for $\mathcal{D}(\mathring{H}^{(\tau)})$ takes the form of \cite[eq.~(I.1.1.27)]{albeverio-solvable}. In particular, the analysis of the extension parameter $\alpha$ done in \cite[Chapter I.1]{albeverio-solvable} shows that the two-body point interaction modelled by the Hamiltonian $H_\alpha$ has $s$-wave scattering length equal to $-(4\pi\alpha)^{-1}$.
\end{remark}

\begin{remark}{Remark}
The function $\phi$ in the decompositions \eqref{eq:D_HringTau-xversion} and \eqref{eq:Hring-tau_form1_xversion} is customarily referred to as the ``\emph{regular part}'' of the given $g$ of the operator domain or the form domain of $\mathring{H}^{(\tau)}$, the difference $g-\phi$ taking the name of the ``\emph{singular part}'' of $g$. Formulas \eqref{eq:D_HringTau-xversion} and \eqref{eq:Hring-tau_form1_xversion}  give, for a generic 
$g=\phi+\phi(0) (8\pi\sqrt{\lambda})/\tau\,G_\lambda\in\mathcal{D}(\mathring{H}^{(\tau)})$,
\[
\langle\, g,(\mathring{H}^{(\tau)}+\lambda\mathbbm{1})\,g\rangle\;=\;\langle\, \phi,(-\Delta+\lambda\mathbbm{1})\,\phi\rangle + \frac{8\pi\sqrt{\lambda}}{\tau}\,|\phi(0)|^2
\]
which provides once more the interpretation of the point-like character of the interaction modelled by $\mathring{H}^{(\tau)}$ at $x=0$.
\end{remark}

\begin{remark}{Remark}
Owing to the decomposition \eqref{eq:D_HringTau-xversion}, a generic $g\in\mathcal{D}(\mathring{H}^{(\tau)})$ displays the characteristic asymptotics of TMS-type (see \eqref{eq:BP-generic-xi} above) when $x\to 0$. Indeed, using the continuity of $\phi$, one derives from \eqref{eq:D_HringTau-xversion}
\[
g(x)\;=\;\phi(0)\,\frac{2\sqrt{\lambda}}{\tau}\Big(\frac{1}{|x|}+\frac{\tau}{2\sqrt{\lambda}}\Big)+o(1)\qquad\textrm{ as }x\to 0\,,
\]
and replacing $\phi(0)$ with $\xi$ according to \eqref{eq:phi0_xi} and $\tau$ with $\alpha$ given by \eqref{eq:alpha-tau} one obtains
\begin{equation*}
g(x)\;=\;\xi\sqrt{\frac{\pi}{2}\,}\,\Big(\frac{1}{|x|}+4\pi\alpha+\sqrt{\lambda}\Big)+o(1)\qquad\textrm{ as }x\to 0\,.
\end{equation*}
In the expression above $\lambda>0$ is arbitrary (and at fixed $g$ the charge $\xi$ is implicitly $\lambda$-dependent, see \eqref{eq:D_HringTau}) and one can therefore read the asymptotics as $\lambda\to 0$; thus, in terms of the scattering length $a=-(4\pi\alpha)^{-1}$ of the interaction, one has
\begin{equation}\label{eq:TMS_recovered_1+1}
g(x)\;=\;\xi\sqrt{\frac{\pi}{2}\,}\Big(\frac{1}{|x|}-\frac{1}{a}\Big)+o(1)\qquad\textrm{ as }x\to 0\,,
\end{equation}
\end{remark}
which has the form of the TMS asymptotics \eqref{eq:BP-generic-xi}.

\begin{remark}{Remark}\label{rem:renormalis1}
Another customary and equivalent expression for the action of $\mathring{H}^{(\tau)}$ in spatial coordinates is obtained by taking the inverse Fourier transform in \eqref{eq:Hring-tau_action2} (while the inverse Fourier transform in \eqref{eq:Hring-tau_action1} yielded \eqref{eq:Hring-tau_action1-xversion}). One finds
\begin{equation}\label{eq:Hring-tau_action1-xversion2}
\mathring{H}^{(\tau)}g\;=\;-\Delta g-(2\pi)^{3/2}\xi\,\delta\,,\qquad g\in\mathcal{D}(\mathring{H}^{(\tau)})\,,
\end{equation}
where $\delta$ is the Dirac distribution. The l.h.s.~of \eqref{eq:Hring-tau_action1-xversion2} is an $L^2$-function for each $g\in\mathcal{D}(\mathring{H}^{(\tau)})$ and the r.h.s.~expresses this $L^2$-function as the difference of two distributions. In general an element $g\in\mathcal{D}(\mathring{H}^{(\tau)})$ does not belong to $H^2(\mathbb{R}^3)$, in which case $-\Delta g$ is only meant as a distributional derivative: this term has an $L^2$-part plus a distributional (non-square-integrable) part $(2\pi)^{3/2}\xi\,\delta$ which is \emph{cancelled} in the difference in the r.h.s.~of  \eqref{eq:Hring-tau_action1-xversion2}. This is consistent with \eqref{eq:TMS_recovered_1+1} above, where formally (as $x\to 0$) one obtains a distributional contribution in $-\Delta g$ given by $-\Delta |x|^{-1}$, which is precisely a $\delta$-distribution. The cancellation occurring in the r.h.s.~of  \eqref{eq:Hring-tau_action1-xversion2} can be regarded as the \emph{renormalisation} of $-\Delta g$ needed to give meaning to $\mathring{H}^{(\tau)}g$ when $g\in\mathcal{D}(\mathring{H}^{(\tau)})$. In Remark \ref{rem:renormalis2} below we shall complete this comment by showing that \eqref{eq:Hring-tau_action1-xversion2} expresses the very same normalisation \eqref{eq:Berezin-Faddeev-1} announced by Berezin and Faddeev.
\end{remark}

\subsection{Ter-Martirosyan--Skornyakov construction for the point interaction}\label{subsec:TMS-1+1}

The functions in $\mathcal{D}(\mathring{H}^*)$ have the following asymptotic behaviour.

\begin{lemma}{Lemma}\label{lemma:asymptotic_integral}
Let $g$ be an arbitrary function in $\mathcal{D}(\mathring{H}^*)$. For a fixed $\lambda>0$ let $\widehat{g}=\widehat{f}+(p^2+\lambda)^{-2}\eta+(p^2+\lambda)^{-1}\xi$ be the decomposition of $g$ obtained in Proposition \ref{prop:D_Hdostar} 
for some $f\in\mathcal{D}(\mathring{H})$ and some $\eta,\xi\in\mathbb{C}$.
Then
\begin{equation}\label{eq:g_*_asymptotics}
\int_{\substack{ \\ \,p\in\mathbb{R}^3 \\ \! |p|<R}}{\:\widehat g(p) \,\ud p}\;=\;4\pi\xi R+\Big(\!-2\pi^2\sqrt{\lambda}\,\xi+\frac{\;\pi^2}{\sqrt{\lambda}}\,\eta\Big)+o(1)\qquad\textrm{as}\qquad R\to +\infty\,.
\end{equation}
\end{lemma}

\begin{proof}
Owing to Lemma \ref{lemma:dom_Hdot_Ftransform}, equation \eqref{eq:actionHdot_to_f}, $\int_{|p|<R}\widehat{f}(p)\,\ud p\to 0$ as $R\to\infty$, so this integral is a $o(1)$ contribution in \eqref{eq:g_*_asymptotics}. As for the two other summands in the decomposition of $\widehat{g}$, we have 
\begin{equation}\label{eq:Regularpart0}
\begin{split}
\int_{\substack{\,p\in\mathbb{R}^3 \\ \! |p|<R}}{\frac{\eta}{(p^2+\lambda)^2} \,\ud p}&\;=\;4\pi\eta \int_{0}^R{\frac{\rho^2}{(\rho^2+\lambda)^2} \,\ud \rho}\;=\;-\frac{2\pi\eta R}{R^2+\lambda}+\frac{2\pi\eta\,}{\sqrt{\lambda}}\arctan\Big(\frac{R}{\sqrt{\lambda}}\Big) \\
&\;=\;\frac{\;\pi^2}{\sqrt{\lambda}}\,\eta+o(1)\qquad\textrm{as}\qquad R\to +\infty
\end{split}
\end{equation}
and
\begin{equation}
\begin{split}
\int_{\substack{\,p\in\mathbb{R}^3 \\ \! |p|<R}}{\frac{\xi}{p^2+\lambda} \,\ud p}&\;=\;4\pi\xi \int_{0}^R{\frac{\rho^2}{\rho^2+\lambda} \,\ud \rho}\;=\;4\pi\xi \Big(R-\sqrt{\lambda}\arctan\Big(\frac{R}{\sqrt{\lambda}}\Big)\Big)  \\
&\;=\;4\pi\xi R-2\pi^2\sqrt{\lambda}\,\xi +o(1)\qquad\textrm{as}\qquad R\to +\infty\,,
\end{split}
\end{equation}
which complete the computation of the r.h.s.~of \eqref{eq:g_*_asymptotics}.
\end{proof}

One \emph{defines} the TMS extension $\mathring{H}_{\langle\alpha\rangle}$ of $\mathring{H}$, $\alpha\in\mathbb{R}\cup\{\infty\}$, to be the restriction of $\mathring{H}^*$ to those $g$'s of $\mathcal{D}(\mathring{H}^*)$ for which in the asymptotics \eqref{eq:g_*_asymptotics} the coefficients of the $O(R)$ term and of the $O(1)$ term are proportional by a factor $\alpha$, more precisely
\begin{equation}\label{eq:TMS_cond_alpha_xi_eta_1}
8\pi^3\alpha\,\xi\;=\;-2\pi^2\sqrt{\lambda}\,\xi+\frac{\;\pi^2}{\sqrt{\lambda}}\,\eta\,.
\end{equation}
Thus, $\mathcal{D}(\mathring{H}_{\langle\alpha\rangle})$ consists of all $g$'s of $\mathcal{D}(\mathring{H}^*)$ for which
\begin{equation}\label{eq:TMS_cond_asymptotics_1}
\int_{\substack{\,p\in\mathbb{R}^3 \\ \! |p|<R}}\,{\widehat g(p) \,\ud p}\;=\;4\pi\xi \,(R+2\pi^2\alpha)+o(1) \qquad\textrm{as}\qquad R\to +\infty\,.
\end{equation}
Both \eqref{eq:TMS_cond_alpha_xi_eta_1} and \eqref{eq:TMS_cond_asymptotics_1} express a \emph{TMS condition}; it  is a constraint on the singular part of the $g$'s of $\mathcal{D}(\mathring{H}^*)$ that restricts the choice of $\xi,\eta\in\mathbb{C}$ to those such that
\begin{equation}\label{eq:eta-xi-1}
\eta\;=\;2\sqrt{\lambda}\,(4\pi\alpha+\sqrt{\lambda})\,\xi\,.
\end{equation}

Condition \eqref{eq:eta-xi-1} above is precisely of the form $\eta=\tau\xi$ with
\begin{equation}\label{eq:tau-alpha}
\tau\;=\;2\sqrt{\lambda}\,(4\pi\alpha+\sqrt{\lambda})\,.
\end{equation}
Owing to \eqref{eq:D_HringTau} of Theorem \ref{thm:1p1_extensions_Birman}, 
this implies at once that the operator $\mathring{H}_{\langle\alpha\rangle}$ selected by the TMS condition \eqref{eq:TMS_cond_alpha_xi_eta_1}-\eqref{eq:TMS_cond_asymptotics_1}-\eqref{eq:eta-xi-1} above is precisely the self-adjoint extension $\mathring{H}^{(\tau)}$ of $\mathring{H}$ qualified by the parameter $\tau$ given by \eqref{eq:tau-alpha} in terms of $\alpha$. We also observe that \eqref{eq:tau-alpha} is the inverse of \eqref{eq:alpha-tau}: therefore, by what recalled in Remark \ref{rem:Htau_Halpha}, the TMS extension $\mathring{H}_{\langle\alpha\rangle}$ models a two-body point interaction with $s$-wave scattering length equal to $-(4\pi\alpha)^{-1}$.

Thus, 
\begin{itemize}
 \item imposing the TMS condition on $\mathring{H}^*$ with parameter $\alpha$ produces the self-adjoint extension of $\mathring{H}$ that gives the point interaction with scattering length $-(4\pi\alpha)^{-1}$;
 \item the collection of all the TMS extensions obtained this way cover the whole family of self-adjoint extensions of $\mathring{H}$: \emph{all self-adjoint extensions of $\mathring{H}$ are of TMS type}. 
\end{itemize}

\begin{remark}{Remark}\label{rem:renormalis2}
With the analysis of this Subsection we can supplement Remark \ref{rem:renormalis1} above with the following observation. First, one recognises that the condition \eqref{eq:Berezin-Faddeev-2}  identified by Berezin and Faddeev for each self-adjoint extension of $\mathring{H}$ is precisely of the TMS form \eqref{eq:TMS_cond_asymptotics_1}. Furthermore, the asymptotics \eqref{eq:TMS_cond_asymptotics_1} shows that for each $g\in\mathcal{D}(\mathring{H}_{\langle\alpha\rangle})$ the corresponding charge $\xi$ is given by
\begin{equation}
\xi\;=\;\lim_{R\to+\infty}\frac{1}{4\pi R}\int_{\substack{\,p\in\mathbb{R}^3 \\ \! |p|<R}}\,{\widehat g(p) \,\ud p}
\end{equation}
which plugged into \eqref{eq:Hring-tau_action2} yields
\begin{equation}
\widehat{(\mathring{H}_{\langle\alpha\rangle} g)}(p)\;=\;p^2 \widehat{g}(p)-\lim_{R\to+\infty}\frac{1}{4\pi R}\int_{\substack{\,p\in\mathbb{R}^3 \\ \! |p|<R}}\,{\widehat g(p) \,\ud p}\,.
\end{equation}
This is exactly the property \eqref{eq:Berezin-Faddeev-1} announced by Berezin and Faddeev, and it is the Fourier-transformed version of the identity \eqref{eq:Hring-tau_action1-xversion2} that expresses $\mathring{H}_{\langle\alpha\rangle} g$ by means of a suitable renormalisation of $-\Delta g$.
\end{remark}

\section{TMS Hamiltonians for the three-body problem with point interaction}\label{sec:2+1}

As discussed in Section \ref{sec:history_TMS}, the problem of a three-particle quantum system with two-body point interaction has been studied since long. However, one still has a relatively limited knowledge of the corresponding Hamiltonians, primarily their self-adjoint realisation (whereas the information on their stability and spectral properties is only partial). The novel difficulty, as opposed to the two-body case, is due to the fact that the Ter-Martirosyan--Skornyakov condition does not select in general a domain of self-adjointness, and one has to further study the self-adjoint extension of the resulting TMS operator.

We restrict our discussion to the most studied case, that of two identical fermions in interaction with a third particle of different nature (the so called ``\emph{2+1 fermionic system}'').

\subsection{The ``2+1 fermionic system''}

After removing the centre of mass, the free Hamiltonian of a three-dimensional system of two identical fermions of unit mass in relative positions $x_1,x_2$ with respect to a third particle of different species and with mass $m$ is the operator $-\Delta_{x_1}-\Delta_{x_2}-\frac{2}{m+1}\nabla_{x_1}\cdot\nabla_{x_2}$ acting on the Hilbert space
\begin{equation}
\cH\;=\;L^2_\mathrm{f}(\mathbb{R}^3\times\mathbb{R}^3,\ud x_1\ud x_2)\,,
\end{equation}
the subscript `f' standing for the fermionic sector of the $L^2$-space, i.e, the square-integrable functions that are anti-symmetric under exchange $x_1\leftrightarrow x_2$. Following the same path as in Subsection \ref{subsec:2delta}, one therefore starts with the  operator
\begin{equation}\label{eq:Hring_2+1}
\begin{split}
\mathring{H}\;&:=\;-\Delta_{x_1}-\Delta_{x_2}-\frac{2}{m+1}\nabla_{x_1}\cdot\nabla_{x_2} \\
\mathcal{D}(\mathring{H})\;&:=\;H^2_0((\mathbb{R}^3\times\mathbb{R}^3)\!\setminus\!(\Gamma_1\cup\Gamma_2))\cap\cH\,,
\end{split}
\end{equation}
where 
\begin{equation}
\Gamma_j\;:=\;\{(x_1,x_2)\in\mathbb{R}^3\times\mathbb{R}^3\,|\,x_j=0\}\,,\qquad j=1,2,
\end{equation}
and 
\begin{equation}
H^2_0((\mathbb{R}^3\times\mathbb{R}^3)\!\setminus\!(\Gamma_1\cup\Gamma_2))\;=\;\overline{C^\infty_0((\mathbb{R}^3\times\mathbb{R}^3)\!\setminus\!(\Gamma_1\cup\Gamma_2))}^{\|\,\|_{H^2}}\,.
\end{equation}

$\mathring{H}$ is a densely defined, closed, positive, and symmetric operator on $\cH$.\ As such, $\mathring{H}$ has equal deficiency indices, and in Proposition \ref{prop:D_Hdostar_2+1} below we shall see that they are infinite. Any self-adjoint extension of $\mathring{H}$ has a natural interpretation of Hamiltonian of point interaction between each fermion and the third particle.

It is convenient first to characterise $\mathcal{D}(\mathring{H})$ in Fourier transform.

\begin{lemma}{Lemma}\label{lemma:dom_Hdot_Ftransform_2+1}~
\begin{itemize}
 \item[(i)] The domain of $\mathring{H}$ is given by
 \begin{equation}\label{eq:dom_H_dot}
 \begin{split}
 \mathcal{D}(\mathring{H})\;&=\;\left\{f\in H^2_\mathrm{f}(\mathbb{R}^3\times\mathbb{R}^3)\left|\!
 \begin{array}{c}
 \iint\widehat{f}(p_1,p_2)\,\widehat{\eta}(p_1)\,\ud p_1\ud p_2=0 \\
 \forall\eta\in H^{-1/2}(\mathbb{R}^3)
 \end{array}\!\!\right.\right\} \\
 &=\;\big\{\,f\in H^2_\mathrm{f}(\mathbb{R}^3\times\mathbb{R}^3)\,\big|\,\langle\widehat{f},h\rangle=0\;\;\forall h\in\mathcal{X}\,\} \,,
 \end{split}
 \end{equation}
 where
 \begin{equation}
 \mathcal{X}\;:=\;\big\{\,h\,|\,h(p_1,p_2)=\widehat{\eta}(p_1)-\widehat{\eta}(p_2)\textrm{ for some }\eta\in H^{-1/2}(\mathbb{R}^3)\,\big\}
 \end{equation}
 and the duality product is meant  between the spaces $L^2(\mathbb{R}^3\times\mathbb{R}^3,(1+p_1^2+p_2^2)^2\ud p_1\ud p_2)$ and $L^2(\mathbb{R}^3\times\mathbb{R}^3,(1+p_1^2+p_2^2)^{-2}\ud p_1\ud p_2)$. The action of $\mathring{H}$ is given by
 \begin{equation}\label{eq:Hring_ftransf}
 \widehat{(\mathring{H}f)}(p_1,p_2)\;=\;(p_1^2+p_2^2+\mu\,p_1\cdot p_2)\,\widehat{f}(p_1,p_2)\,,
 \end{equation}
 where
 \begin{equation}\label{eq:mu}
  \mu\;:=\;\frac{2}{m+1}\,.
 \end{equation}
 In other words, $\mathcal{D}(\mathring{H})$ consists of the $H^2_\mathrm{f}$-functions $f$ such that $\int_{\mathbb{R}^3}\widehat{f}(p_1,p_2)\ud p_j=0$ \linebreak in $H^{1/2}(\mathbb{R}^3)$, where $j\in\{1,2\}$, and $\mathring{H}$ acts on such $f$'s as the free two-body (negative) Laplacian.
 \item[(ii)] The Friedrichs extension of $\mathring{H}$ is given by
 \begin{equation}\label{eq:Friedrichs_2+1}
 \mathcal{D}(\mathring{H}_F)\;=\;H^2_\mathrm{f}(\mathbb{R}^3\times \mathbb{R}^3)\,,\qquad \widehat{(\mathring{H}_Ff)}(p_1,p_2)\;=\;(p_1^2+p_2^2+\mu\,p_1\cdot p_2)\,\widehat{f}(p_1,p_2)\,.
 \end{equation}
\end{itemize}
\end{lemma}

\begin{proof}
Clearly, for $j\in\{1,2\}$ one has $\Gamma_j=\cap_{\nu=1}^3\Gamma_{j,\nu}$, where $x_j\equiv(x_{j,1},x_{j,2},x_{j,3})\in\mathbb{R}^3$ and  $\Gamma_{j,\nu}$ is the hyperplane $x_{j,\nu}=0$. For a generic $f\in H^2(\mathbb{R}^3\times\mathbb{R}^3)$ a standard trace theorem \cite[Lemma 16.1]{Trtar-SobSpaces_interp} 
asserts that $f|_{\Gamma_{1,1}}\in H^{2-\frac{1}{2}}(\mathbb{R}^2\times\mathbb{R}^3)$, and if in addition
$f\in H^2_0((\mathbb{R}^3\times\mathbb{R}^3)\!\setminus\!\Gamma_{1,\nu})$, then $f|_{\Gamma_{1,1}}=0$.
Thus, by repeated application of the trace theorem to a function  $f\in H^2_0((\mathbb{R}^3\times\mathbb{R}^3)\!\setminus\!\Gamma_{1})$, one finds $f|_{\Gamma_1}=0$ in $H^{2-\frac{3}{2}}(\mathbb{R}^3)$. Summarising,
\begin{equation}\label{eq:trace_on_hyperplanes}
f\in\mathcal{D}(\mathring{H})\qquad\Leftrightarrow\qquad f|_{\Gamma_j}\in H^{\frac{1}{2}}(\mathbb{R}^3)\quad\textrm{and}\quad f|_{\Gamma_j}=0\,,\quad j\in\{1,2\}\,.
\end{equation}
One therefore has that $f\in\mathcal{D}(\mathring{H})$ is equivalent to
\begin{equation}\label{eq:finDHring_iff}
0\;=\;\langle \eta,f|_{\Gamma_2}\rangle_{H^{-\frac{1}{2}},H^{\frac{1}{2}}}\;=\;\int_{\mathbb{R}^3}\overline{\widehat{\eta}(p_1)}\,(\widehat{f|_{\Gamma_2})}(p_1)\,\ud p_1\qquad\forall\eta\in H^{-\frac{1}{2}}(\mathbb{R}^3)\,.
\end{equation}
Let now $f\in\mathcal{D}(\mathring{H})$. The Fourier transforms that follow are all of $L^2$-functions, therefore the corresponding integral expressions are meant as $L^2$-norm limits. 
From 
\[
f|_{\Gamma_2}(x_1)\;=\;f(x_1,0)\;=\;\frac{1}{\;(2\pi)^{3/2}}\iint_{\mathbb{R}^3\times\mathbb{R}^3}\widehat{f}(p_1,p_2) \,e^{\ii x_1 p_1} \ud p_1\,\ud p_2
\]
for a.e.~$x_1\in\mathbb{R}^3$, and from the distributional identity $\delta(k)=(2\pi)^{-3}\int_{\mathbb{R}^3}\ud x\,e^{\ii x k}$, one finds
\[
\begin{split}
\widehat{(f|_{\Gamma_2})}(p_1)\;&=\;\frac{1}{\;(2\pi)^{3/2}}\int_{\mathbb{R}^3}f|_{\Gamma_2}(x_1)\,e^{-\ii p_1 x_1}\,\ud x_1 \\
&=\;\frac{1}{\;(2\pi)^3}\int_{\mathbb{R}^3}\ud x_1 \,e^{\ii x_1 (q_1-p_1)}\iint_{\mathbb{R}^3\times\mathbb{R}^3}\widehat{f}(q_1,q_2) \,\ud q_1\,\ud q_2 \;=\;\int_{\mathbb{R}^3}\widehat{f}(p_1,p_2)\,\ud p_2\,.
\end{split}
\]
From this and from \eqref{eq:trace_on_hyperplanes} one deduces that the function $p_1\mapsto\int\widehat{f}(p_1,p_2)\,\ud p_2$ vanishes in $H^{1/2}(\mathbb{R}^3)$.
Plugging the last identity into \eqref{eq:finDHring_iff} yields
\[
\iint_{\mathbb{R}^3\times\mathbb{R}^3}\widehat{f}(p_1,p_2)\,\widehat{\eta}(p_1)\,\ud p_1\ud p_2\;=\;0 \qquad\forall\eta\in H^{-1/2}(\mathbb{R}^3)\,,
\]
which proves the first line in \eqref{eq:dom_H_dot}.
This result, together with the anti-symmetry of $f$, implies that
\begin{equation}\label{eq:duality_product_for_f_2+1}
\iint_{\mathbb{R}^3\times\mathbb{R}^3}\widehat{f}(p_1,p_2)\,(\widehat{\eta}(p_1)-\widehat{\eta}(p_2))\,\ud p_1\ud p_2\;=\;0 \quad \forall\eta\in H^{-1/2}(\mathbb{R}^3)\quad (f\in \mathcal{D}(\mathring{H}))\,.
\end{equation}
We now observe that $\widehat{f}\in L^2(\mathbb{R}^3\times\mathbb{R}^3,(1+p_1^2+p_2^2)^2\ud p_1\ud p_2)$ because $\mathcal{D}(\mathring{H})\subset H^2(\mathbb{R}^3\times\mathbb{R}^3)$, and that the map $(p_1,p_2)\mapsto\widehat{\eta}(p_1)-\widehat{\eta}(p_2)$ belongs to $L^2(\mathbb{R}^3\times\mathbb{R}^3,(1+p_1^2+p_2^2)^{-2}\ud p_1\ud p_2)$, because
\[
\begin{split}
\iint_{\mathbb{R}^3\times\mathbb{R}^3}&\frac{|\widehat{\eta}(p_1)-\widehat{\eta}(p_2)|^2}{(1+p_1^2+p_2^2)^2}\,\ud p_1\ud p_2\;\leqslant\;4\iint_{\mathbb{R}^3\times\mathbb{R}^3}\frac{|\widehat{\eta}(p_1)|^2}{(1+p_1^2+p_2^2)^2}\,\ud p_1\ud p_2 \\
&=4\iint_{\mathbb{R}^3\times\mathbb{R}^3}\frac{|\widehat{\eta}(p_1)|^2}{(p_1^2+1)^{1/2}}\,\frac{(p_1^2+1)^{1/2}}{(1+p_1^2+p_2^2)^2}\,\ud p_1\ud p_2\;\leqslant\;4\pi^2\|\eta\|_{H^{-1/2}}^2\;<\;+\infty\,,
\end{split}
\]
where we used 
\[
\int_{\mathbb{R}^3}\frac{\ud p_2}{(1+p_1^2+p_2^2)^2}\;=\;\frac{\pi^2}{\:(p_1^2+1)^{1/2}}\,.
\]
Thus, one can regard \eqref{eq:duality_product_for_f_2+1} as the vanishing of a duality product between the spaces $L^2(\mathbb{R}^3\times\mathbb{R}^3,(1+p_1^2+p_2^2)^2\ud p_1\ud p_2)$ and $L^2(\mathbb{R}^3\times\mathbb{R}^3,(1+p_1^2+p_2^2)^{-2}\ud p_1\ud p_2)$, and since $f\in\mathcal{D}(\mathring{H})$ was arbitrary, one concludes the second line in \eqref{eq:dom_H_dot}. Equations \eqref{eq:Hring_ftransf}-\eqref{eq:mu}  are the Fourier-transformed version of \eqref{eq:Hring_2+1}, and this concludes the proof of part (i).
Concerning part (ii), we first observe that the form domain of $\mathring H$, namely the completion of $H^2_0((\mathbb{R}^3\times\mathbb{R}^3)\!\setminus\!(\Gamma_1\cup\Gamma_2))\cap\cH$, is the space
\begin{equation}
D[\mathring H]\;=\;H^1_0((\mathbb{R}^3\times\mathbb{R}^3)\!\setminus\!(\Gamma_1\cup\Gamma_2))\cap\cH\,.
\end{equation}
By a standard approximation with smooth and compactly supported functions (see Appendix \ref{app:inclusions}),
\begin{equation}\label{useful_inclusion_2+1}
H^2_{\mathrm{f}}(\mathbb{R}^3\times\mathbb{R}^3)\;\subset\; H^1_0((\mathbb{R}^3\times\mathbb{R}^3)\!\setminus\!(\Gamma_1\cup\Gamma_2))\cap\cH\;=\;H^1_{\mathrm{f}}(\mathbb{R}^3\times\mathbb{R}^3)\,,
\end{equation}
and $H^2_{\mathrm{f}}(\mathbb{R}^3\times\mathbb{R}^3)$ is the domain of the self-adjoint extension of $\mathring H$ given by the free negative Laplacian on $\cH$. Therefore, owing to the characterisation of the Friedrichs extension $\mathring{H}_F$ as the unique self-adjoint extension of $\mathring{H}$ whose operator domain is contained in $\mathcal{D}[\mathring{H}]$, one deduces immediately \eqref{eq:Friedrichs_2+1}.
\end{proof}

In the next Proposition we characterise the adjoint of $\mathring{H}$, which is the preliminary step in
order to identify the self-adjoint extensions of $\mathring{H}$ within the general scheme of the Kre{\u\i}n-Vi\v{s}ik-Birman theory.

\begin{theorem}{Proposition}\label{prop:D_Hdostar_2+1}
Let $\lambda>0$.
\begin{itemize}
 \item[(i)] One has
 \begin{equation}\label{eq:ker_hring*_2+1}
 \begin{split}
 \!\!\ker (\mathring{H}^*+\lambda\mathbbm{1})\;&=\;\left\{ u_\xi\in L^2_\mathrm{f}(\mathbb{R}^3\!\times\!\mathbb{R}^3)\left|\!
 \begin{array}{l}
 \widehat{u}_\xi(p_1,p_2)=\displaystyle\frac{\widehat{\xi}(p_1)-\widehat{\xi}(p_2)}{p_1^2+p_2^2+\mu\,p_1\cdot p_2+\lambda} \\
 \xi\in H^{-1/2}(\mathbb{R}^3)
 \end{array}\!\!\right.\right\}
 \end{split}
 \end{equation}
 \item[(ii)] There exist constants $c_1,c_2>0$ such that for a generic $u_\xi\in\ker (\mathring{H}^*+\lambda\mathbbm{1})$ one has
 \begin{equation}\label{eq:uxi-equivalent-norms}
  c_1\|\xi\|_{H^{-1/2}(\mathbb{R}^3)}\leqslant\;\|u_\xi\|_{\cH}\;\leqslant c_2\|\xi\|_{H^{-1/2}(\mathbb{R}^3)}\,.
 \end{equation}
 \item[(iii)] The domain and the action of the adjoint of $\mathring{H}$ are given by
 \begin{equation}\label{eq:decompositionD_Hdostar_2+1}
  \mathcal{D}(\mathring{H}^*)\;=\;\left\{\!\!
  \begin{array}{c}
  g\in L^2_\mathrm{f}(\mathbb{R}^3\!\times\!\mathbb{R}^3)\quad\textrm{such that} \\
  \widehat{g}(p_1,p_2)=\displaystyle\widehat{f}(p_1,p_2)+\frac{\widehat{u}_\eta(p_1,p_2)}{p_1^2+p_2^2+\mu\,p_1\cdot p_2+\lambda}+\widehat{u}_\xi(p_1,p_2) \\
  \textrm{for}\quad f\in\mathcal{D}(\mathring{H})\,,\quad \eta,\xi\in H^{-1/2}(\mathbb{R}^3)
  \end{array}\!\!\right\} 
 \end{equation}
 and
 \begin{eqnarray}
  (\widehat{(\mathring{H}^*+\lambda) g)}\,(p_1,p_2)\!\!\!&=&\!\!\! (p_1^2+p_2^2+\mu\,p_1\cdot p_2+\lambda) \,\widehat{F}_\lambda(p_1,p_2) \label{eq:actionHdotstar_to_g-f_2+1}\\
  \widehat{(\mathring{H}^* g)}(p_1,p_2)\!\!\!&=&\!\!\!(p_1^2+p_2^2+\mu\,p_1\cdot p_2) \widehat{g}(p_1,p_2)-(\widehat{\xi}(p_1)-\widehat{\xi}(p_2)), \label{eq:actionHdotstar_to_g-f_2+1bis}
 \end{eqnarray}
 where $u_\eta$ and $u_\xi$ are defined as in \eqref{eq:ker_hring*_2+1} above, and
\begin{equation}
\widehat{F}_\lambda(p_1,p_2)\;:=\;\widehat{f}(p_1,p_2)+\frac{\widehat{u}_\eta(p_1,p_2)}{p_1^2+p_2^2+\mu\,p_1\cdot p_2+\lambda}\,.
\end{equation}
\end{itemize}
\end{theorem}

\begin{proof}
In order to prove part (i) we use the fact that $u\in \ker(\mathring H^*+\lambda\mathbbm{1})=\ran(\mathring H+\lambda\mathbbm{1})^{\perp}$ if and only if for every $f\in \mathcal{D}(\mathring H+\lambda\mathbbm{1})$
\begin{equation*}
\begin{split}
0\;&=\;\int_{\mathbb{R}^3}{u(x_1,x_2)((\mathring{H}+\lambda\mathbbm{1})f)(x_1,x_2)\,\ud x_1 \ud x_2} \\
&=\;\int_{\mathbb{R}^3}{\widehat u(p_1,p_2)((\mathring{H}+\lambda\mathbbm{1})f)^{\widehat{\;}}(p_1,p_2)\,\ud p_1 \ud p_2} \\
&=\; \int_{\mathbb{R}^3}{\widehat u(p_1,p_2)(p_1^2+p_2^2+\mu p_1\cdot p_2+\lambda)\widehat f(p_1,p_2)\,\ud p_1 \ud p_2}\,,
\end{split}
\end{equation*}
where we applied \eqref{eq:Hring_ftransf} in the last step. Since $\mu\in(0,2)$ (owing to \eqref{eq:mu} with $m> 0$) and $\lambda>0$, then
\begin{equation}\label{eq:lambda-equiv-1}
(1+p_1^2+p_2^2)\;\sim\;(p_1^2+p_2^2+\mu p_1\cdot p_2+\lambda)
\end{equation}
(in the sense that each quantity controls the other from above and from below), and hence the fact that $u\in \ker(\mathring H^*+\lambda\mathbbm{1})\subset L^2(\mathbb{R}^3\times\mathbb{R}^3,\ud p_1\ud p_2)$ is equivalent to 
\[
 \big(\,(p_1,p_2)\mapsto(p_1^2+p_2^2+\mu p_1\cdot p_2+\lambda)\,\widehat{u}(p_1,p_2)\,\big)\;\in\; L^2(\mathbb{R}^3\times\mathbb{R}^3,(1+p_1^2+p_2^2)^{-2}\ud p_1\ud p_2)\,.
\]
Therefore, the last identity above implies, owing to the second line of \eqref{eq:dom_H_dot}, that 
\begin{equation*}
(p_1^2+p_2^2+\mu p_1\cdot p_2+\lambda)\,\widehat{u}(p_1,p_2)\;=\;\widehat\xi(p_1)-\widehat\xi(p_2)
\end{equation*}
for some $\xi\in H^{-1/2}(\mathbb{R}^3)$, and hence each $u\in \ker(\mathring H^*+\lambda\mathbbm{1})$ is of the form $u\equiv u_\xi$ given by \eqref{eq:ker_hring*_2+1}. 
Part (ii) is taken directly from \cite[Lemma B.2]{CDFMT-2015}.
Concerning part (iii), because of $(\mathring{H}+\lambda\mathbbm{1})_F=\mathring{H}_F+\lambda\mathbbm{1}$ and \eqref{eq:Friedrichs_2+1}, we have that
\begin{equation}\label{eq:HringF-2}
((\mathring{H}+\lambda\mathbbm{1})_F^{-1}u)^{\widehat{\;}}\;=\;(p^2+\lambda)^{-1}\widehat{u}\,.
\end{equation}
This, together with  the decomposition formula \eqref{eq:DomS*} of Lemma \ref{lemma:krein_decomp_formula} and the characterisation  \eqref{eq:ker_hring*_2+1} of $\ker (\mathring{H}^*+\lambda\mathbbm{1})$ yield immediately \eqref{eq:decompositionD_Hdostar_2+1}. The decomposition \eqref{eq:DomS*} also implies that the action of $\mathring{H}^*+\lambda\mathbbm{1}$ on a generic element $f+(\mathring{H}^*+\lambda\mathbbm{1})_F^{-1}u_\eta+u_\xi\in\mathcal{D}(\mathring{H}^*+\lambda\mathbbm{1})$ is the same as the action of $\mathring{H}_F+\lambda\mathbbm{1}$ on the component $f+(\mathring{H}^*+\lambda\mathbbm{1})_F^{-1}u_\eta\in\mathcal{D}(\mathring{H}_F+\lambda\mathbbm{1})$, while $(\mathring{H}^*+\lambda\mathbbm{1}) u_\xi=0$: this is precisely  \eqref{eq:actionHdotstar_to_g-f_2+1}. As for \eqref{eq:actionHdotstar_to_g-f_2+1bis}, it is an immediate consequence of  \eqref{eq:decompositionD_Hdostar_2+1} and \eqref{eq:actionHdotstar_to_g-f_2+1}.
\end{proof}

\subsection{General scheme for self-adjoint realisations of the 2+1 fermionic model}

The space $\ker (\mathring{H}^*+\lambda\mathbbm{1})$ determined in \eqref{eq:ker_hring*_2+1} is the Kre{\u\i}n space of the model (the ``boundary value space'', in modern terminology). It is known by the Kre{\u\i}n-Vi\v{s}ik-Birman theory that the self-adjoint extensions of  $\mathring{H}$ defined in \eqref{eq:Hring_2+1} are parametrised by self-adjoint operators acting on Hilbert subspaces of $\ker (\mathring{H}^*+\lambda\mathbbm{1})$, according to the classification given by Theorem \ref{thm:VB-representaton-theorem_Tversion} in Appendix \ref{app:KVB}.

Unlike the two-body model discussed in the previous Section, where $\ker (\mathring{H}^*+\lambda\mathbbm{1})$ was one-dimensional, in the three-body model this space is \emph{infinite-dimensional} (compare \eqref{eq:ker_hring*_2+1} with \eqref{eq:ker_hring*}), which makes the variety of the self-adjoint extensions of  $\mathring{H}$ much more complicated. In this respect, the extensions of Ter-Martirosyan--Skornyakov type form a proper sub-family (in the two-body case, all self-adjoint extensions were of TMS type).

Eventually one introduces the TMS extensions of $\mathring{H}$, analogously to Section \ref{subsec:TMS-1+1} for the two-body model, as restrictions of $\mathring{H}^*$ to domains characterised by special asymptotics of their wave-functions in the vicinity of the coincidence hyperplanes. Prior to that, in this Subsection we shall develop a general scheme for the identification of a generic self-adjoint extension of $\mathring{H}$, within which we will later select those of TMS type.

To this aim, it is convenient first to define the expressions
\begin{equation}\label{eq:Tlambda}
\widehat{(T_\lambda\,\xi)}(p)\;:=\;2\pi^2\sqrt{\nu p^2+\lambda}\;\widehat{\xi}(p)+\int_{\mathbb{R}^3}\frac{\widehat{\xi}(q)}{p^2+q^2+\mu p\cdot q+\lambda}\,\ud q
\end{equation}
and
\begin{equation}\label{eq:Wlambda}
\widehat{(W_\lambda\,\xi)}(p)\;:=\;\frac{2\pi^2}{\sqrt{\nu p^2+\lambda}\,}\,\widehat{\xi}(p)-2\!\int_{\mathbb{R}^3}\frac{\widehat{\xi}(q)}{(p^2+q^2+\mu p\cdot q+\lambda)^2}\,\ud q
\end{equation}
for fixed $\lambda>0$, where $\mu$ is given by \eqref{eq:mu} and 
\begin{equation}\label{eq:nu}
\nu\;:=\;1-\frac{\:\mu^2}{4}\;=\;\frac{m(m+2)}{(m+1)^2}\,.
\end{equation}

Since, for arbitrary $\varepsilon>0$ and $\xi\in H^{-\frac{1}{2}+\varepsilon}(\mathbb{R}^3)$,
\[
\begin{split}
\Big|\int_{\mathbb{R}^3}\frac{\widehat{\xi}(q)}{p^2+q^2+\mu p\cdot q+\lambda}\,\ud q\Big|\;\leqslant\;\|\xi\|_{H^{-\frac{1}{2}+\varepsilon}}\Big(\int_{\mathbb{R}^3}\frac{(q^2+1)^{\frac{1}{2}-\varepsilon}}{(p^2+q^2+\mu p\cdot q+\lambda)^2}\,\ud q\Big)^{1/2}\;<\;+\infty
\end{split}
\]
(owing to a Schwartz inequality in the first step and \eqref{eq:lambda-equiv-1} in the second one), we see that the integral in \eqref{eq:Tlambda} is finite for any $\xi\in H^{-\frac{1}{2}+\varepsilon}(\mathbb{R}^3)$, $\varepsilon>0$, while in general it diverges when $\xi\in H^{-\frac{1}{2}}(\mathbb{R}^3)$, as the example $\widehat{\xi}_0(q)\;:=\;\mathbf{1}_{\{|q|\geqslant 2\}}(|q|\ln|q|)^{-1}$ shows. A similar argument shows that the integral in \eqref{eq:Wlambda} is finite too at least for $\xi\in H^{-\frac{1}{2}}(\mathbb{R}^3)$. Summarising, $(\widehat{T_\lambda\,\xi})(p)$ is well-defined point-wise for almost every $p\in\mathbb{R}^3$ for $\xi\in H^{-\frac{1}{2}+\varepsilon}(\mathbb{R}^3)$, $\varepsilon>0$, whereas $(\widehat{W_\lambda\,\xi})(p)$ is so (at least) for $\xi\in H^{-1/2}(\mathbb{R}^3)$.

The relevance of functions of the form $\widehat{T_\lambda\,\xi}$ and $\widehat{W_\lambda\,\eta}$ is due to the fact that they arise in the asymptotic behaviour of the elements of $\mathcal{D}(\mathring{H}^*)$.

\begin{lemma}{Lemma}\label{lemma:asymptotic_integral_2+1}
Let $g$ be an arbitrary function in $\mathcal{D}(\mathring{H}^*)$. For a fixed $\lambda>0$ consider the decomposition of  $\widehat{g}$ in terms of $\widehat{f}$, $\widehat{u}_\xi$, $\widehat{u}_\eta$ given by \eqref{eq:decompositionD_Hdostar_2+1}. 
Then, as $R\to +\infty$,
\begin{equation}\label{eq:g_*_asymptotics_2+1}
\begin{split}
\int_{\substack{ \\ \,p_2\in\mathbb{R}^3 \\ \! |p_2|<R}}{\:\widehat{g}(p_1,p_2) \,\ud p_2}\;&=\;4\pi\widehat{\xi}(p_1) R+\big(-\widehat{(T_\lambda\,\xi)}(p_1)+\!\!\begin{array}{c}\frac{1}{2}\end{array}\!\!\!\widehat{(W_\lambda\,\eta)}(p_1)\big)+o(1)
\end{split}
\end{equation}
as a point-wise identity for almost every $p_1$.
\end{lemma}

\begin{remark}{Remark}
For a generic $g\in \mathcal{D}(\mathring{H}^*)$, and correspondingly for a generic charge $\xi\in H^{-\frac{1}{2}}(\mathbb{R}^3)$, the quantity in the l.h.s.~of \eqref{eq:g_*_asymptotics_2+1} is \emph{infinite} for every finite $R$ because, as remarked after the definitions \eqref{eq:Tlambda}-\eqref{eq:Wlambda}, the quantity $(\widehat{T_\lambda\,\xi})(p)$ is in general infinite when $\xi\in H^{-\frac{1}{2}}(\mathbb{R}^3)$. Instead, when additionally $\xi\in H^{-\frac{1}{2}+\varepsilon}(\mathbb{R}^3)$, with $\varepsilon>0$, the r.h.s.~of \eqref{eq:g_*_asymptotics_2+1} is finite (for almost every $p_1\in\mathbb{R}^3$): this case corresponds to a dense set of $g$'s in $\mathcal{D}(\mathring{H}^*)$, and for such $g$'s the quantity in the l.h.s.~of \eqref{eq:g_*_asymptotics_2+1} is \emph{finite} for finite $R$ and only diverges, linearly in $R$, as $R\to +\infty$.
\end{remark}

\begin{proof}[Proof of Lemma \ref{lemma:asymptotic_integral_2+1}]
Splitting $\widehat g$ according to \eqref{eq:decompositionD_Hdostar_2+1} yields
\begin{equation}\label{eq:gasympt_2+1}
\int_{\substack{ \\ \,p_2\in\mathbb{R}^3 \\ \! |p_2|<R}}{\:\widehat{g} \,\ud p_2}\;=\int_{\substack{ \\ \,p_2\in\mathbb{R}^3 \\ \! |p_2|<R}}{\:\widehat{f} \,\ud p_2}+\int_{\substack{ \\ \,p_2\in\mathbb{R}^3 \\ \! |p_2|<R}}{\:\widehat{u}_{\xi} \,\ud p_2}+\int_{\substack{ \\ \,p_2\in\mathbb{R}^3 \\ \! |p_2|<R}}{\:\frac{\widehat{u}_{\eta}}{p_1^2+p_2^2+\mu p_1\cdot p_2 +\lambda} \,\ud p_2}\,.
\end{equation}
The first summand in the r.h.s.~of \eqref{eq:gasympt_2+1}  is a $o(1)$-contribution, owing to Lemma \ref{lemma:dom_Hdot_Ftransform_2+1}(i). 
For the second summand, which is re-written as
\begin{align*}
\int_{\substack{ \\ \,p_2\in\mathbb{R}^3 \\ \! |p_2|<R}}&{\:\widehat{u}_{\xi}(p_1,p_2) \,\ud p_2}\;=
\\
&=\;\widehat{\xi}(p_1)\int_{\substack{ \\ \,p_2\in\mathbb{R}^3 \\ \! |p_2|<R}}\;\frac{\ud p_2}{p_1^2+p_2^2+\mu p_1\cdot p_2 +\lambda}-\int_{\substack{ \\ \,p_2\in\mathbb{R}^3 \\ \! |p_2|<R}}\,{\frac{\widehat{\xi}(p_2)}{p_1^2+p_2^2+\mu p_1 \cdot p_2+\lambda}\,\ud p_2}\,,
\end{align*}
one finds
\[
\begin{split}
\int_{\substack{ \\ \,p_2\in\mathbb{R}^3 \\ \! |p_2|<R}}&\;\frac{\ud p_2}{p_1^2+p_2^2+\mu p_1\cdot p_2 +\lambda}\;=\;2\pi\int_0^Rs^2\,\ud s\int_{-1}^1\frac{\ud y}{p_1^2+s^2+\mu|p_1|sy+\lambda} \\
&=\;\frac{2\pi}{\mu|p_1|}\int_0^Rs\ln\frac{s^2+p_1^2+\mu|p_1|s+\lambda}{s^2+p_1^2-\mu|p_1|s+\lambda}\,\ud s \\
&=\;2\pi R\,\Big(1+\frac{R}{2\mu|p_1|}\ln\frac{R^2+p_1^2+\mu|p_1|R+\lambda}{R^2+p_1^2-\mu|p_1| R+\lambda}\Big) \\
&\qquad\quad+2\pi\sqrt{\nu p_1^2+\lambda}\,\Big(\!\arctan\frac{\mu |p_1|-2R}{2\sqrt{\nu p_1^2+\lambda}}-\arctan\frac{\mu |p_1|+2R}{2\sqrt{\nu p_1^2+\lambda}}\,\Big) \\
&\qquad\quad+\pi\frac{(4\nu-2)p_1^2+\lambda}{4\sqrt{\nu p_1^2+\lambda}}\,\ln\frac{R^2+p_1^2+\mu|p_1|R+\lambda}{R^2+p_1^2-\mu|p_1| R+\lambda} \\
&=\;4\pi R-2\pi^2\sqrt{\nu p_1^2+\lambda}+o(1)\,,
\end{split}
\]
whence
\[
\begin{split}
\int_{\!\!\!\!\!\!\!\!\!\!\substack{\\ \\ \\ \,p_2\in\mathbb{R}^3 \\ \! |p_2|<R}}{\:\widehat{u}_{\xi}(p_1,p_2) \,\ud p_2}\;&=\;4\pi\,\widehat{\xi}(p_1)R-2\pi^2\widehat{\xi}(p_1)\sqrt{\nu p_1^2+\lambda}-\!\int_{\mathbb{R}^3}{\frac{\widehat{\xi}(p_2)}{p_1^2+p_2^2+\mu p_1 \cdot p_2+\lambda}\ud p_2}+o(1)\,.
\end{split}
\]
Analogously, for the third summand in the r.h.s.~of \eqref{eq:gasympt_2+1} one has
\begin{align*}
\int_{\substack{ \\ \,p_2\in\mathbb{R}^3 \\ \! |p_2|<R}}\,&{\:\frac{\widehat{u}_{\eta}(p_1,p_2)}{p_1^2+p_2^2+\mu p_1 \cdot p_2+\lambda} \,\ud p_2}\;=
\\
&=\;\widehat{\eta}(p_1)\int_{\substack{ \\ \,p_2\in\mathbb{R}^3 \\ \! |p_2|<R}}\;\frac{\ud p_2}{(p_1^2+p_2^2+\mu p_1\cdot p_2 +\lambda)^2}-\int_{\substack{ \\ \,p_2\in\mathbb{R}^3 \\ \! |p_2|<R}}\,{\frac{\widehat{\eta}(p_2)}{(p_1^2+p_2^2+\mu p_1 \cdot p_2+\lambda)^2}\,\ud p_2}
\end{align*}
and
\begin{equation}\label{eq_integral_in_lemma}
\begin{split}
\int_{\mathbb{R}^3}&\;\frac{\ud p_2}{(p_1^2+p_2^2+\mu p_1\cdot p_2 +\lambda)^2}\;=\;2\pi\int_0^{+\infty}\!\!s^2\,\ud s\int_{-1}^1\frac{\ud y}{(p_1^2+s^2+\mu|p_1|sy+\lambda)^2} \\
&=\;\int_0^{+\infty}\!\!\frac{4\pi s^2}{(p_1^2+s^2+\lambda)^2-\mu^2p_1^2 s^2}\,\ud s\;=\;\frac{\pi^2}{\sqrt{\nu p_1^2+\lambda}}\,,
\end{split}
\end{equation}
whence
\[
\int_{\!\!\!\!\!\!\!\!\!\!\substack{\\ \\ \\ \,p_2\in\mathbb{R}^3 \\ \! |p_2|<R}}\!{\frac{\widehat{u}_{\eta}(p_1,p_2)}{p_1^2+p_2^2+\mu p_1 \cdot p_2+\lambda} \,\ud p_2}\;=\frac{\pi^2\,\widehat{\eta}(p_1)}{\sqrt{\nu p_1^2+\lambda}}-\int_{\mathbb{R}^3}\!{\frac{\widehat{\eta}(p_2)}{(p_1^2+p_2^2+\mu p_1 \cdot p_2+\lambda)^2}\,\ud p_2}+o(1)\,.
\]
These findings, re-written with the definitions \eqref{eq:Tlambda} and \eqref{eq:Wlambda}, show that the r.h.s.~of \eqref{eq:gasympt_2+1} is precisely given by formula \eqref{eq:g_*_asymptotics_2+1}.
\end{proof}

Thus, functions in $\mathcal{D}(\mathring{H}^*)$ display completely analogous asymptotics to the two-body model. It is important to observe, however, that whereas \eqref{eq:g_*_asymptotics} was an identity between \emph{scalars}, here \eqref{eq:g_*_asymptotics_2+1} is a \emph{point-wise} almost everywhere identity between \emph{functions}. This is a crucial difference to keep into account when one imposes the Ter-Martirosyan--Skornyakov condition in such asymptotics.

To elaborate on this point further at a later stage, let us also interpret $\xi\mapsto T_\lambda\xi$ and $\xi\mapsto W_\lambda\xi$  as maps between suitable functional spaces.



\begin{theorem}{Proposition}\label{prop:T-W}
Let $\lambda>0$. 
\begin{itemize}
 \item[(i)] For each $s\geqslant 1$ the expression \eqref{eq:Tlambda} defines a densely defined and symmetric operator  $T_\lambda:\mathcal{D}(T_\lambda)\subset L^2(\mathbb{R}^3)\to L^2(\mathbb{R}^3)$ with domain $\mathcal{D}(T_\lambda):=H^s(\mathbb{R}^3)$. Moreover, $T_\lambda$ maps continuously $H^s(\mathbb{R}^3)$ into $H^{s-1}(\mathbb{R}^3)$ for each $s\in(-\frac{1}{2},\frac{3}{2})$. Instead, $T_\lambda H^{3/2}(\mathbb{R}^3)\nsubseteq H^{1/2}(\mathbb{R}^3)$.
 \item[(ii)] The expression \eqref{eq:Wlambda} defines a bounded, positive, and invertible linear operator $W_\lambda:H^{-1/2}(\mathbb{R}^3)\to H^{1/2}(\mathbb{R}^3)$, and for generic $u_\xi,u_\eta\in\ker (\mathring{H}^*+\lambda\mathbbm{1})$ one has
 \begin{equation}\label{eq:scalar_products}
 \langle u_\xi,u_\eta\rangle_{\cH}\;=\;\langle \xi,W_\lambda\eta\rangle_{H^{-\frac{1}{2}}(\mathbb{R}^3),H^{\frac{1}{2}}(\mathbb{R}^3)}\,.
 \end{equation}
\end{itemize}
\end{theorem}

\begin{remark}{Remark}
The choice of $L^2(\mathbb{R}^3)$ as the Hilbert space where to study $T_\lambda$ is made here for consistency with the previous literature \cite{Minlos-Shermatov-1989,Menlikov-Minlos-1991,Menlikov-Minlos-1991-bis,Shermatov-2003,Minlos-2011-preprint_May_2010,Minlos-2010-bis,Minlos-2012-preprint_30sett2011,Minlos-2012-preprint_1nov2012,Minlos-RusMathSurv-2014}, but it has no fundamental reason. As we shall discuss in Subsection \ref{subsec:TMS_2p1} below, what is intrinsically fundamental for the self-adjoint extension theory of $\mathring{H}$ is the operator $W_\lambda^{-1}T_\lambda$ on the Hilbert space $H^{-1/2}(\mathbb{R}^3)$.
\end{remark}

\begin{proof}[Proof of Proposition \ref{prop:T-W}]
(i) We re-write \eqref{eq:Tlambda} as $T_\lambda=L_\lambda+Q_\lambda$, where
\begin{equation}\label{eq:T-L-Q}
\widehat{(L_\lambda\xi})(p)\;:=\;2\pi^2\sqrt{\nu p^2+\lambda}\;\widehat{\xi}(p)\,,\qquad \widehat{(Q_\lambda\xi)}(p)\;:=\;\int_{\mathbb{R}^3}\frac{\widehat{\xi}(q)}{p^2+q^2+\mu p\cdot q+\lambda}\,\ud q\,.
\end{equation}
The symmetry of $T_\lambda$ on $L^2(\mathbb{R}^3)$ is obvious, since $L_\lambda$ is the multiplication by a real function and $Q_\lambda$ is  an integral operator with real and symmetric kernel, and so too is the fact that $\mathcal{D}(T_\lambda)$ is dense in $L^2(\mathbb{R}^3)$. It is also clear that $\|L_\lambda\xi\|_{H^{s-1}}\sim\|\xi\|_{H^s}$, thus it only remains to prove that $\|Q_\lambda\xi\|_{H^{s-1}}\lesssim\|\xi\|_{H^s}$, i.e., $\|(1+p^2)^{\frac{s-1}{2}}\widehat{(Q_\lambda\xi)}\|_2\lesssim\|(1+p^2)^{\frac{s}{2}}\widehat{\xi}\|_2$. In turn, setting $h(p):=(1+p^2)^{\frac{s}{2}}\widehat{\xi}(p)$, the last inequality is equivalent to $\|\widetilde{Q}_\lambda h\|_2\lesssim\|h\|_2$, where
\[
(\widetilde{Q}_\lambda h)(p)\;:=\;\int_{\mathbb{R}^3}K_\lambda(p,q)\,h(q)\,\ud q\,,\qquad K_\lambda(p,q)\;:=\;\frac{\;(1+p^2)^{\frac{s-1}{2}}}{(p^2+q^2+\mu p\cdot q+\lambda)(1+q^2)^{\frac{s}{2}}}\,.
\]
It is easily verified (using \eqref{eq:lambda-equiv-1} to introduce the $(p^2+q^2+1)$-factors)  that for the \emph{positive} function
\[
f(p)\;:=\;(1+p^2)^{-\frac{3}{4}}
\]
one has
\[\tag{*}
\begin{split}
\int_{\mathbb{R}^3}K_\lambda(p,q)f(p)\,\ud p\;&\lesssim\;\frac{1}{(1+q^2)^{\frac{s}{2}}}\int_{\mathbb{R}^3}\frac{\ud p}{(p^2+q^2+1)(1+p^2)^{\frac{5}{4}-\frac{s}{2}}} \\
&\lesssim\;\frac{1}{(1+q^2)^{\frac{s}{2}}}\,\frac{1}{(1+q^2)^{\frac{3}{4}-\frac{s}{2}}}\;=\;f(q)\,,\qquad{\textstyle s\in(-\frac{1}{2},\frac{3}{2})}\,,
\end{split}
\]
and 
\[\tag{**}
\begin{split}
\int_{\mathbb{R}^3}K_\lambda(p,q)f(q)\,\ud q\;&\lesssim\;(1+p^2)^{\frac{s-1}{2}}\int_{\mathbb{R}^3}\frac{\ud q}{(p^2+q^2+1)(1+q^2)^{\frac{3}{4}+\frac{s}{2}}} \\
&\lesssim\;(1+p^2)^{\frac{s-1}{2}}\,\frac{1}{(1+p^2)^{\frac{1}{4}+\frac{s}{2}}}\;=\;f(p)\,,\qquad \qquad{\textstyle s\in(-\frac{1}{2},\frac{3}{2})}\,.
\end{split}
\]
A standard Schur test based on (*) and (**) implies $\|\widetilde{Q}_\lambda h\|_2\lesssim\|h\|_2$ and hence $\|Q_\lambda\xi\|_{H^{s-1}}\lesssim\|\xi\|_{H^s}$ for $s\in(-\frac{1}{2},\frac{3}{2})$ and arbitrary $\xi\in H^{s}(\mathbb{R}^3)$.  The function $\widehat{\xi}_0:=\mathbf{1}_{\{|p|\leqslant 1\}}$ is a counter-example showing that the same bound cannot hold for $s\geqslant \frac{3}{2}$: indeed, clearly $\xi_0\in H^{s}(\mathbb{R}^3)$ for arbitrary $s\in\mathbb{R}$, but
\[
\widehat{(Q_\lambda\xi_0)}(p)\;\sim\;\int_{\substack{ \\ \,q\in\mathbb{R}^3 \\ \! |q|\leqslant 1}}\frac{\ud q}{p^2+q^2+1}\;\sim\;\frac{1}{\;(1+p^2)}
\]
and hence $Q_\lambda\xi_0\notin H^{1/2}(\mathbb{R}^3)$.

(ii) For arbitrary $\xi,\eta\in H^{-1/2}(\mathbb{R}^3)$ we compute
\begin{equation*} 
\begin{split}
\langle u_{\xi}, u_{\eta}\rangle_\cH\;&=\;\int_{\mathbb{R}^3\times\mathbb{R}^3}{\frac{\overline{\,\widehat{\xi}(p_1)}-\overline{\,\widehat{\xi}(p_2)}}{p_1^2+p_2^2+\mu p_1\cdot p_2+\lambda}\:\frac{\widehat \eta (p_1)-\widehat \eta (p_2)}{p_1^2+p_2^2+\mu p_1\cdot p_2+\lambda}\,\ud p_1\ud p_2} \\
&=\;2\int_{\mathbb{R}^3\times\mathbb{R}^3}{\frac{\overline{\,\widehat{\xi}(p_1)}\,\widehat\eta(p_1)}{(p_1^2+p_2^2+\mu p_1\cdot p_2+\lambda)^2}-\frac{\overline{\,\widehat{\xi}(p_1)}\,\widehat \eta(p_2)}{(p_1^2+p_2^2+\mu p_1\cdot p_2+\lambda)^2}\,\ud p_1 \ud p_2} \\
&=\;2\int_{\mathbb{R}^3}{\overline{\,\widehat{\xi}(p_1)}\,\Big ( \frac{\pi^2}{\sqrt{\nu p_1^2+\lambda}}\,\widehat \eta (p_1)-\int_{\mathbb{R}^3}{\frac{\widehat \eta (p_2)}{(p_1^2+p_2^2+\mu p_1\cdot p_2+\lambda)^2}\,\ud p_2} \Big )\ud p_1} \\
&=\;\int_{\mathbb{R}^3}\overline{\widehat{\xi}(p)}\,\widehat{(W_\lambda\xi)}(p)\,\ud p\,,
\end{split}
\end{equation*}
where we used the symmetry under exchange $p_1\leftrightarrow p_2$ in the second step, \eqref{eq_integral_in_lemma} in the third step, and \eqref{eq:Wlambda} in the last step. Therefore,
\[
\begin{split}
\|W_\lambda\eta\|_{H^{1/2}}\;&=\;\sup_{\|\xi\|_{H^{-1/2}=1}}\Big|\int_{\mathbb{R}^3}\overline{\widehat{\xi}(p)}\,\widehat{(W_\lambda\eta)}(p)\,\ud p\,\Big|\;=\;\sup_{\|\xi\|_{H^{-1/2}=1}}\big|\langle u_\xi,u_\eta\rangle_{\cH}\big| \\
&\leqslant\;\sup_{\|\xi\|_{H^{-1/2}=1}}\|u_\xi\|_{\cH}\|u_\eta\|_{\cH}\;\leqslant \;\textrm{const}\cdot\|\eta\|_{H^{-1/2}}\qquad\forall\eta\in H^{-1/2}(\mathbb{R}^3)\,,
\end{split}
\]
where we used \eqref{eq:uxi-equivalent-norms} in the last step, which shows that $W_\lambda H^{-1/2}(\mathbb{R}^3)\subset H^{1/2}(\mathbb{R}^3)$, that the map $W_\lambda:H^{-1/2}(\mathbb{R}^3)\to H^{1/2}(\mathbb{R}^3)$ is bounded, and that \eqref{eq:scalar_products} holds true. 
Owing to \eqref{eq:scalar_products}, one has $\langle \eta,W_\lambda\eta\rangle_{H^{-1/2},H^{1/2}}=\|u_\eta\|_{\cH}^2\geqslant 0$, thus $W_\lambda$ is positive. Furthermore, the following chain of implications holds: $W_\lambda\eta=0$ $\Rightarrow$ $\langle u_\xi,u_\eta\rangle_{\cH}=0$ $\forall u_\xi\in\ker (\mathring{H}^*+\lambda\mathbbm{1})$ $\Rightarrow$ $u_\eta=0$ $\Rightarrow$ $\eta=0$, where we used  \eqref{eq:scalar_products} in the first implication and \eqref{eq:ker_hring*_2+1} in the last one; this proves that $W_\lambda$ is injective and hence invertible on its range. For a generic $\xi\in H^{-1/2}(\mathbb{R}^3)$ one has this chain of implications: $\langle \xi,\phi\rangle_{H^{-1/2},H^{1/2}}=0$ $\forall\phi=W_\lambda\eta\in \ran\,W_\lambda$ 
$\Rightarrow$ $\langle u_\xi,u_\eta\rangle_{\cH}=0$ $\forall u_\eta\in\ker (\mathring{H}^*+\lambda\mathbbm{1})$  $\Rightarrow$ $u_\xi=0$ $\Rightarrow$ $\xi=0$ (again using \eqref{eq:scalar_products} in the first implication and \eqref{eq:ker_hring*_2+1} in the last one), hence by duality $\ran\,W_\lambda$ must be dense in $H^{1/2}(\mathbb{R}^3)$. Since $W_\lambda$ is bounded, then $\ran\,W_\lambda$ is also closed in $H^{1/2}(\mathbb{R}^3)$, thus $W_\lambda$ is an invertible bijection $H^{-1/2}(\mathbb{R}^3)\to H^{1/2}(\mathbb{R}^3)$.
\end{proof}

As an immediate consequence of Proposition \ref{prop:T-W},
\begin{equation}\label{eq:W-scalar-product}
\langle \xi,\eta\rangle_{W_\lambda}\;:=\;\langle \xi,W_\lambda\,\eta\rangle_{H^{-\frac{1}{2}},H^{\frac{1}{2}}}\;=\;\langle u_\xi,u_\eta\rangle_{\cH}
\end{equation}
defines a scalar product in $H^{-\frac{1}{2}}(\mathbb{R}^3)$. It is \emph{equivalent} to the standard scalar product of $H^{-\frac{1}{2}}(\mathbb{R}^3)$, as follows by combining \eqref{eq:W-scalar-product} with \eqref{eq:uxi-equivalent-norms}.

We shall denote by $H^{-1/2}_{W_\lambda}(\mathbb{R}^3)$ the Hilbert space consisting of the $H^{-\frac{1}{2}}(\mathbb{R}^3)$-functions and equipped with the scalar product $\langle\cdot,\cdot\rangle_{W_\lambda}$.
Then the map
\begin{equation}\label{eq:isomorphism_Ulambda}
\begin{split}
U_\lambda\,:\,\ker (\mathring{H}^*+\lambda\mathbbm{1})\;&\;\xrightarrow[]{\;\;\;\cong\;\;\;}\;H^{-1/2}_{W_\lambda}(\mathbb{R}^3)\,,\qquad u_\xi \longmapsto \,\xi
\end{split}
\end{equation}
is an isomorphism between Hilbert spaces, with $\ker (\mathring{H}^*+\lambda\mathbbm{1})$ equipped with  the standard scalar product inherited from $\cH$.

One can therefore equivalently parametrise the self-adjoint extensions of $\mathring{H}$ in terms of self-adjoint operators acting on Hilbert subspaces of $\ker (\mathring{H}^*+\lambda\mathbbm{1})$ or of $H^{-1/2}_{W_\lambda}(\mathbb{R}^3)$. The whole family of such extensions is given by the  Kre{\u\i}n-Vi\v{s}ik-Birman theory through the classification of Theorem \ref{thm:VB-representaton-theorem_Tversion}.

\subsection{Ter-Martirosyan--Skornyakov Hamiltonians of $2+1$ point interaction}\label{subsec:TMS_2p1}

The previous analysis brings us now to the class of operators on $H^{-1/2}_{W_\lambda}(\mathbb{R}^3)$ (or, also, operators on $\ker (\mathring{H}^*+\lambda\mathbbm{1})$) which identify those self-adjoint extensions of $\mathring{H}$ of Ter-Martirosyan--Skornyakov type.

As it will emerge in the following, the crucial point is the possibility of reducing  $T_\lambda$ to a $L^2$-closed invariant subspace with values into $H^{1/2}(\mathbb{R}^3)$ and to define on the orthogonal complement another $H^{1/2}$-valued symmetric operator $S_0$. Given the resulting $\mathcal{T}_\lambda=S_0\oplus T_\lambda$ one has then to investigate the self-adjointness of $W_\lambda^{-1}\mathcal{T}_\lambda$ on  $H^{-1/2}_{W_\lambda}(\mathbb{R}^3)$. We observe that this is related to, but it is \emph{not} the same question of the self-adjointness of $T_\lambda$ on $L^2(\mathbb{R}^3)$.

The study of the self-adjoint extensions of $T_\lambda$, as a densely defined and symmetric operator on $L^2(\mathbb{R}^3)$, has been carried on systematically in a series of works by
Minlos and Shermatov \cite{Minlos-Shermatov-1989}, Melnikov and Minlos \cite{Menlikov-Minlos-1991,Menlikov-Minlos-1991-bis}, Shermatov \cite{Shermatov-2003}, and Minlos \cite{Minlos-2011-preprint_May_2010,Minlos-2010-bis,Minlos-2012-preprint_30sett2011,Minlos-2012-preprint_1nov2012,Minlos-RusMathSurv-2014}. In the additional work \cite{CDFMT-2012} by one of us and co-workers, the Friedrichs extension of $T_\lambda$ was studied (in those regime of masses $m$ in which $T_\lambda$ itself is semi-bounded below). It is relevant to remark that in all those works $\mathcal{D}(T_\lambda)$ was fixed in spaces of various regularity, at least $H^1(\mathbb{R}^3)$. We recall from Proposition \ref{prop:T-W}(i) that $T_\lambda H^1(\mathbb{R}^3)\subset L^2(\mathbb{R}^3)$ because both its multiplicative part $L_\lambda$ and its integral part $Q_\lambda$ map separately $H^1(\mathbb{R}^3)$ into $L^2(\mathbb{R}^3)$: as discussed in \cite{CDFMT-2012}, in the domain of a self-adjoint extension $\widetilde{T}_\lambda$ of $T_\lambda$ there are elements $\xi$ for which neither $L_\lambda\xi$ nor $Q_\lambda\xi$ is square-integrable, but their difference is, due to a cancellation of singularities in $L_\lambda\xi+Q_\lambda\xi$.

More precisely, $T_\lambda$ commutes with the rotations in $\mathbb{R}^3$ and, with respect to the canonical decomposition
\begin{equation}\label{eq:L2_ell_decomposition}
L^2(\mathbb{R}^3)\;\cong\;\bigoplus_{\ell=0}^\infty L^2(\mathbb{R}^+,r^2\,\ud r)\otimes\mathrm{span}\{Y_{\ell,-\ell},\dots,Y_{\ell,\ell}\}\;\equiv\;\bigoplus_{\ell=0}^\infty L^2_\ell(\mathbb{R}^3)
\end{equation}
(where the $Y_{\ell,m}$'s are the spherical harmonics on $\mathbb{S}^2$), $T_\lambda$ leaves each $L^2_\ell(\mathbb{R}^3)$ invariant and is densely defined and symmetric on $L^2_\ell(\mathbb{R}^3)$, thus
\begin{equation}\label{eq:T-Tell}
T_\lambda\;=\;\bigoplus_{\ell=0}^\infty \,T_\lambda^{(\ell)}\qquad\qquad T_\lambda^{(\ell)}\,\textrm{ symmetric on } L^2_\ell(\mathbb{R}^3)
\end{equation}
and 
\begin{equation}\label{eq:Tlambda-rad-ang}
T_\lambda^{(\ell)}\;=\;\mathcal{T}_{\lambda}^{(\ell)}\otimes \mathbbm{1}\qquad\textrm{ on }\qquad L^2_\ell(\mathbb{R}^3)\;\cong\;L^2(\mathbb{R}^+,r^2\,\ud r)\otimes\mathrm{span}\{Y_{\ell,-\ell},\dots,Y_{\ell,\ell}\}\,.
\end{equation}
Therefore, the study of the self-adjointness or of the self-adjoint extensions of $T_\lambda$ boils down to the same study for each $T_\lambda^{(\ell)}$. It is today well-known from the works cited above that for \emph{even} $\ell$'s  $T_\lambda^{(\ell)}$ is self-adjoint on $L^2_\ell(\mathbb{R}^3)$, while for \emph{odd} $\ell$'s there exist masses $m_1>m_3>m_5>\cdots$ such that $T_\lambda^{(\ell)}$ is self-adjoint on $L^2_\ell(\mathbb{R}^3)$ for $m>m_\ell$ and it has instead a one-parameter family of self-adjoint extensions for $m\in(0,m_\ell]$. 

As seen in the proof of Proposition \ref{prop:T-W}(i), $T_\lambda$ fails to map $H^{3/2}(\mathbb{R}^3)$ into $H^{1/2}(\mathbb{R}^3)$ and the counter-example considered therein was a function in $L^2_{\ell=0}(\mathbb{R}^3)\cap H^{3/2}(\mathbb{R}^3)$. In fact, that failure is exceptional and it does not occur for $H^{3/2}$-functions with a sufficient amount of oscillations, as the following Proposition shows.

\begin{theorem}{Proposition}\label{prop:T-ell}
For each $\ell\geqslant 1$, and in terms of the notation of \eqref{eq:L2_ell_decomposition}-\eqref{eq:T-Tell},
\begin{equation}\label{eq:T-ell-3/2-1/2}
\|T_\lambda\xi\|_{H^{1/2}}\;\lesssim\|\xi\|_{H^{3/2}}\qquad \forall \xi\in H^{3/2}(\mathbb{R}^3)\cap L^2_\ell(\mathbb{R}^3)\,,
\end{equation}
whence, in particular,
\begin{equation}
 T_\lambda^{(\ell)}\big(H^{3/2}(\mathbb{R}^3)\cap L^2_\ell(\mathbb{R}^3)\big)\;\subset\; \big(H^{1/2}(\mathbb{R}^3)\cap L^2_\ell(\mathbb{R}^3)\big)\,.
\end{equation}
\end{theorem}

\begin{proof}
As discussed already in the proof of Proposition \ref{prop:T-W}(i), we only need to prove the statement for $Q_\lambda$, the integral part of $T_\lambda$ -- see \eqref{eq:T-L-Q}. Analogously to \eqref{eq:T-Tell}-\eqref{eq:Tlambda-rad-ang},
\begin{equation}\label{eq:Q-Qell}
Q_\lambda\;=\;\bigoplus_{\ell=0}^\infty \,Q_\lambda^{(\ell)}\,,\qquad Q_\lambda^{(\ell)}\;=\;\mathcal{Q}_{\lambda}^{(\ell)}\otimes \mathbbm{1}\quad\textrm{ on }\quad L^2_\ell(\mathbb{R}^3)\,,
\end{equation}
where $\mathcal{Q}_{\lambda}^{(\ell)}$ acts symmetrically on $L^2(\mathbb{R}^+,r^2\,\ud r)$. It is standard to derive from \eqref{eq:T-L-Q} and \eqref{eq:Q-Qell} that 
\begin{equation}\label{eq:kernel qQ_l}
(\mathcal{Q}_{\lambda}^{(\ell)}f)(r)\;=\;2\pi\int_{-1}^{+1}\!\ud y P_\ell(y)\!\int_0^{+\infty}\!\!\frac{f(r')}{r^2+r'^2+\mu r r' y +\lambda}\,r'^2\ud r'\,,
\end{equation}
where
\begin{equation}\label{eq:P-Legendre}
P_\ell(y)\;=\;\frac{1}{2^\ell \ell!}\,\frac{\ud^\ell}{\ud y^\ell}(y^2-1)^\ell
\end{equation}
is the $\ell$-th Legendre polynomial. Thus, proving \eqref{eq:T-ell-3/2-1/2} is equivalent to proving
\begin{equation}
\|(1+r^2)^{\frac{1}{4}}(\mathcal{Q}_{\lambda}^{(\ell)}f)\|_{ L^2(\mathbb{R}^+,r^2\,\ud r)}\;\lesssim\;\|(1+r^2)^{\frac{3}{4}}f\|_{ L^2(\mathbb{R}^+,r^2\,\ud r)}\,,
\end{equation}
which is in turn equivalent to the boundedness in $L^2(\mathbb{R}^+,\ud r)$ of the integral operator $h\mapsto \widetilde{\mathcal{Q}}_{\lambda}^{(\ell)} h$ defined by
\begin{equation}\label{eq:kernel qQtilde_l}
(\widetilde{\mathcal{Q}}_{\lambda}^{(\ell)}h)(r)\;:=\;\int_{-1}^{+1}\!\ud y P_\ell(y)\!\int_0^{+\infty}\!\!\frac{rr'\,(1+r^2)^{\frac{1}{4}}\,h(r')}{(r^2+r'^2+\mu r r' y +\lambda)(1+r'^2)^{\frac{3}{4}}}\,\ud r'\,.
\end{equation}
Using \eqref{eq:P-Legendre} and integrating by parts  $\ell\geqslant 1$ times in $y$ yields
\[
\begin{split}
(\widetilde{\mathcal{Q}}_{\lambda}^{(\ell)}h)(r)\;=\;\frac{(-1)^\ell}{2^\ell\ell!}\int_0^{+\infty}\!\!\!\ud r'\,\frac{rr'\,(1+r^2)^{\frac{1}{4}}h(r')}{(1+r'^2)^{\frac{3}{4}}}\int_{-1}^{+1}\!\ud y\,\frac{(y^2-1)^\ell(\mu r r')^\ell}{(r^2+r'^2+\mu r r' y +\lambda)^{\ell+1}}\,.
\end{split}
\]
Since $|y|\leqslant 1$, analogously to \eqref{eq:lambda-equiv-1}
\begin{equation}\label{eq:lambda-y-equiv-1}
(r^2+r'^2+\mu r r' y +\lambda)\;\sim\;(r_1^2+r_2^2+1)\;\geqslant 0\,.
\end{equation}
Then
\[
\begin{split}
|(\widetilde{\mathcal{Q}}_{\lambda}^{(\ell)}h)(r)|\;&\lesssim\;\int_0^{+\infty}\!\!\!\ud r'\,\frac{rr'\,(1+r^2)^{\frac{1}{4}}\,|h(r')|}{(1+r'^2)^{\frac{3}{4}}}\int_{-1}^{+1}\!\ud y\,\frac{(\mu r r')^\ell}{(r^2+r'^2+\mu r r' y +\lambda)^{\ell+1}} \\
&=\;\int_0^{+\infty}\!\!\!\ud r'\,\frac{rr'\,(1+r^2)^{\frac{1}{4}}\,|h(r')|}{\ell\,(1+r'^2)^{\frac{3}{4}}}\,(\mu r r')^{\ell-1}\,\times \\
&\qquad\qquad\qquad\times\Big(\frac{1}{(r^2+r'^2-\mu r r' y +\lambda)^{\ell}}-\frac{1}{(r^2+r'^2+\mu r r' y +\lambda)^{\ell}}\Big) \\
&\lesssim\;\int_0^{+\infty}\!\!\!\ud r'\,\frac{rr'\,(1+r^2)^{\frac{1}{4}}\,|h(r')|}{(1+r'^2)^{\frac{3}{4}}}\,(\mu r r')^{\ell}\,\times \\
&\qquad\qquad\qquad\times\,\frac{(r^2+r'^2+1)^{\ell-1}}{(r^2+r'^2-\mu r r' y +\lambda)^{\ell}(r^2+r'^2+\mu r r' y +\lambda)^{\ell}} \\
&\lesssim\;\int_0^{+\infty}\!\!\!\ud r'\,\frac{(rr')^{\ell+1}(1+r^2)^{\frac{1}{4}}}{(1+r'^2)^{\frac{3}{4}}(r^2+r'^2+1)^{\ell+1}}\,|h(r')|\;\equiv\;\int_0^{+\infty}\!\!\mathcal{K}_\lambda^{(\ell)}(r,r')\,|h(r')|\,\ud r'\,,
\end{split}
\]
where we used \eqref{eq:lambda-y-equiv-1} in the first and last step, and the formula $(a^\ell-b^\ell)=(a-b)\sum_{j=0}^{n-1}a^{n-j-1}b^{\,j}$ ($a,b\geqslant 0$) in the third step.
From
\[
\int_0^{+\infty}\!\!\frac{r^{\ell+1}(1+r^2)^{\frac{1}{4}}}{(r^2+r'^2+1)^{\ell+1}}\,\ud r\;\lesssim\;(1+r'^2)^{-\frac{2\ell-1}{4}}
\]
we deduce
\[\tag{*}
\sup_{r'>0}\int_0^{+\infty}\!\!\mathcal{K}_\lambda^{(\ell)}(r,r')\,\ud r\;=\;\sup_{r'>0}\frac{r'^{\ell+1}}{(1+r'^2)^{\frac{3}{4}}}\int_0^{+\infty}\!\!\frac{r^{\ell+1}(1+r^2)^{\frac{1}{4}}}{(r^2+r'^2+1)^{\ell+1}}\,\ud r\;\lesssim\;1\,,
\]
and from
\[
\int_0^{+\infty}\!\!\frac{r'^{\ell+1}}{(1+r'^2)^{\frac{3}{4}}(r^2+r'^2+1)^{\ell+1}}\,\ud r'\;\lesssim\;(1+r^2)^{-\frac{2\ell+3}{4}}
\]
we deduce
\[\tag{**}
\begin{split}
\sup_{r>0}\int_0^{+\infty}\!\!&\mathcal{K}_\lambda^{(\ell)}(r,r')\,\ud r'\;= \\
&=\;\sup_{r>0}\:r^{\ell+1}(1+r^2)^{\frac{1}{4}}\!\int_0^{+\infty}\!\!\frac{r'^{\ell+1}}{(1+r'^2)^{\frac{3}{4}}(r^2+r'^2+1)^{\ell+1}}\,\ud r'\;\lesssim\;1\,.
\end{split}
\]
A standard Schur test based on (*) and (**) implies $\|\widetilde{\mathcal{Q}}_{\lambda}^{(\ell)}h\|_2\lesssim\|h\|_2$, thus concluding the proof.
\end{proof}

An immediate consequence of Propositions \ref{prop:T-W}(i) and \ref{prop:T-ell} is the following.

\begin{theorem}{Corollary}
Let $T_\lambda^+$ be the operator acting as  $T_\lambda$ on the Hilbert space $L^2_{+}(\mathbb{R}^3):=\bigoplus_{\ell=1}^\infty L^2_{\ell}(\mathbb{R}^3)$ with domain $\mathcal{D}(T_\lambda^+):=H^{3/2}(\mathbb{R}^3)\cap L^2_{+}(\mathbb{R}^3)$. Then $T_\lambda^+$ is densely defined and symmetric, and it maps continuously $H^{3/2}(\mathbb{R}^3)\cap L^2_{+}(\mathbb{R}^3)$ (with the $H^{\frac{3}{2}}$-norm) into $H^{1/2}(\mathbb{R}^3)\cap L^2_{+}(\mathbb{R}^3)$ (with the $H^{\frac{1}{2}}$-norm).
\end{theorem}

In turn, the Corollary above, together with Proposition, show that although in general
\begin{equation}
T_\lambda H^{3/2}(\mathbb{R}^3)\;\varsupsetneq\;H^{1/2}(\mathbb{R}^3)\,,
\end{equation}
nevertheless the map $W_\lambda^{-1}T_\lambda$ can be defined on parts of $H^{-1/2}(\mathbb{R}^3)$ with values in itself. Elaborating further, we obtain the following result.

\begin{theorem}{Proposition}\label{lem:selfadj_hierarchy}
The following be given: 
\begin{itemize}
 \item two constants $\lambda>0$ and $\alpha\in\mathbb{R}$,
 \item the densely defined and symmetric operator $T_\lambda^+:=\bigoplus_{\ell=1}^\infty \,T_\lambda^{(\ell)}$  on the Hilbert space $L^2_{+}(\mathbb{R}^3):=\bigoplus_{\ell=1}^\infty L^2_{\ell}(\mathbb{R}^3)$ which acts as  $T_\lambda$  with domain $\mathcal{D}(T_\lambda^+):=H^{3/2}(\mathbb{R}^3)\cap L^2_{+}(\mathbb{R}^3)$,
 \item and a densely defined and symmetric operator  $S_0$ on the Hilbert space $L^2_{\ell=0}(\mathbb{R}^3)$ with $\ran S_0\subset H^{1/2}(\mathbb{R}^3)$.
\end{itemize}
   With respect to the decomposition $L^2(\mathbb{R}^3)\cong L^2_{\ell=0}(\mathbb{R}^3)\oplus L^2_+(\mathbb{R}^3)$, let
\begin{equation}\label{eq:defTtildeLambda}
\mathcal{T}_\lambda\;:=\;S_0\oplus T^+_\lambda\,.
\end{equation}
Then $\mathcal{T}_\lambda$ is a densely defined and symmetric operator on $L^2(\mathbb{R}^3)$ and
\begin{equation}\label{eq:Atilde_l_a}
\mathcal{A}_{\lambda,\alpha}\;:=\;2\,W_\lambda^{-1}(\mathcal{T}_\lambda+\alpha\mathbbm{1})\,,\qquad\mathcal{D}(\mathcal{A}_{\lambda,\alpha})\;:=\;\mathcal{D}(\mathcal{T}_\lambda)
\end{equation} 
is a densely defined and symmetric operator on  $H^{-1/2}_{W_\lambda}(\mathbb{R}^3)$. Moreover, if $\widetilde{\mathcal{A}}_{\lambda,\alpha}$ is a self-adjoint extension of $\mathcal{A}_{\lambda,\alpha}$ on $H^{-1/2}_{W_\lambda}(\mathbb{R}^3)$, then 
\begin{equation}\label{eq:unitary_equiv_A_Atilde_2+1}
A_{\lambda,\alpha}\;:=\;U_\lambda^{-1}\widetilde{\mathcal{A}}_{\lambda,\alpha} U_\lambda
\end{equation}
(where $U_\lambda$ is the isomorphism \eqref{eq:isomorphism_Ulambda}) is a self-adjoint operator on $\ker (\mathring{H}^*+\lambda\mathbbm{1})$. 
\end{theorem}


\begin{proof}
The statements for $\mathcal{T}_\lambda$ and $A_{\lambda,\alpha}$ are obvious, and so too is the density of $\mathcal{D}(\mathcal{A}_{\lambda,\alpha})$ in $H^{-1/2}_{W_\lambda}(\mathbb{R}^3)$.
The symmetry of $\mathcal{A}_{\lambda,\alpha}$ follows from the identity, valid for $\eta,\xi\in\mathcal{D}(\mathcal{A}_{\lambda,\alpha})=\mathcal{D}(\mathcal{T}_\lambda)$,
\[
\begin{split}
\!\!\begin{array}{c}\frac{1}{2}\end{array}\!\!\!\langle\eta,\mathcal{A}_{\lambda,\alpha}\xi\rangle_{W_\lambda}\;&=\;\langle\eta,W_\lambda^{-1}(\mathcal{T}_\lambda\xi+\alpha\xi)\rangle_{W_\lambda} \;=\;\langle\eta,(\mathcal{T}_\lambda\xi+\alpha\xi)\rangle_{L^2} \\
&=\;\langle(\mathcal{T}_\lambda\eta+\alpha\eta),\xi\rangle_{L^2}\;=\;\langle W_\lambda W_\lambda^{-1}(\mathcal{T}_\lambda\eta+\alpha\eta),\xi\rangle_{H^{\frac{1}{2}},H^{-\frac{1}{2}}} \\
&=\;\langle W_\lambda^{-1}(\mathcal{T}_\lambda\eta+\alpha\eta),W_\lambda \xi\rangle_{H^{-\frac{1}{2}},H^{\frac{1}{2}}}\;=\;\!\!\begin{array}{c}\frac{1}{2}\end{array}\!\!\!\langle\mathcal{A}_{\lambda,\alpha}\eta,\xi\rangle_{W_\lambda}\,,
\end{split}
\]
where we used the symmetry of $\mathcal{T}_\lambda$ in $L^2(\mathbb{R}^3)$, the fact that $\alpha$ is real, and the properties of $W_\lambda$ discussed in Proposition \ref{prop:T-W}(ii).
\end{proof}

\begin{remark}{Remark}\label{rem:selfadj_of_T_not_enough}
If, in addition to the assumptions of Proposition \ref{lem:selfadj_hierarchy} above, one assumes also that $T_\lambda^+$ and $S_0$ are self-adjoint on their respective Hilbert spaces and hence $\mathcal{T}_\lambda$ is \emph{self-adjoint} on $L^2(\mathbb{R}^3)$, then
\begin{equation}\label{eq:adjoint-not-selfadj}
\mathcal{D}(\mathcal{A}_{\lambda,\alpha}^*)\cap L^2(\mathbb{R}^3)\;=\;\mathcal{D}(\mathcal{A}_{\lambda,\alpha})\qquad\qquad(\mathcal{T}_\lambda=\mathcal{T}_\lambda^*)\,.
\end{equation}
To see this, we use the fact that for generic $\chi\in\mathcal{D}(\mathcal{A}_{\lambda,\alpha}^*)$ there exists  $c_\chi>0$ such that
\[\tag{*}
\big|\langle\chi,\mathcal{A}_{\lambda,\alpha}\xi\rangle_{W_\lambda}\big|\;\leqslant\;c_\chi\,\|\xi\|_{W_\lambda}\qquad\forall\xi\in\mathcal{D}(\mathcal{A}_{\lambda,\alpha})=\mathcal{D}(\mathcal{T}_\lambda)\,.
\]
Owing to \eqref{eq:W-scalar-product}, \eqref{eq:scalar_products}, and \eqref{eq:uxi-equivalent-norms},  
\[
\|\xi\|_{W_\lambda}\;=\;\|u_\xi\|_\cH\;\sim\;\|\xi\|_{H^{-1/2}}\;\leqslant\;\|\xi\|_{L^2}\,,
\]
whereas, for $\chi\in \mathcal{D}(\mathcal{A}_{\lambda,\alpha}^*)\cap L^2(\mathbb{R}^3)$,
\[
\langle\chi,\mathcal{A}_{\lambda,\alpha}\xi\rangle_{W_\lambda}\;=\;2\,\langle\chi,(\mathcal{T}_\lambda\xi+\alpha\xi)\rangle_{L^2}\,.
\]
Therefore, (*) reads
\[
\big|\langle\chi,(\mathcal{T}_\lambda\xi+\alpha\xi)\rangle_{L^2}\big|\;\lesssim\;\|\xi\|_{L^2}\qquad\forall\xi\in\mathcal{D}(\widetilde{\mathcal{T}}_\lambda)
\]
or also (using $|\langle\chi,\xi\rangle_{L^2}|\leqslant\|\chi\|_{L^2}\|\xi\|_{L^2}$)
\[
\big|\langle\chi,\mathcal{T}_\lambda\xi\rangle_{L^2}\big|\;\lesssim\;\|\xi\|_{L^2}\qquad\forall\xi\in\mathcal{D}(\widetilde{\mathcal{T}}_\lambda)\,.
\]
Since $\mathcal{T}_\lambda$ is self-adjoint on $L^2(\mathbb{R}^3)$, the last bound implies $\chi\in\mathcal{D}(\mathcal{T}^*_\lambda)=\mathcal{D}(\mathcal{T}_\lambda)=\mathcal{D}(\mathcal{A}_{\lambda,\alpha})$, whence the conclusion. We observe, however, that this argument and the conclusion \eqref{eq:adjoint-not-selfadj} are not enough to claim that the self-adjointness of $\mathcal{T}_\lambda$ implies the self-adjointness of $\mathcal{A}_{\lambda,\alpha}$: the latter could still have a larger adjoint and admit self-adjoint extensions.
\end{remark}

We thus see that each self-adjoint extension
$\widetilde{\mathcal{A}}_{\lambda,\alpha}$  of $\mathcal{A}_{\lambda,\alpha}$ on $H^{-1/2}_{W_\lambda}(\mathbb{R}^3)$ identifies one self-adjoint extension of $\mathring{H}$ by means of its unitarily equivalent version  $A_{\lambda,\alpha}=U_\lambda^{-1}\widetilde{\mathcal{A}}_{\lambda,\alpha} U_\lambda$.
This extension, call it $\mathring{H}_{\langle\alpha\rangle}$, as prescribed by Theorem \ref{thm:VB-representaton-theorem_Tversion} (namely the Kre{\u\i}n-Vi\v{s}ik-Birman extension theory) is the restriction of $\mathring{H}^*$ to the domain
\begin{equation}\label{eq:D_Halpha_2+1}
\mathcal{D}(\mathring{H}_{\langle\alpha\rangle})\;:=\;\left\{g=f+(\mathring{H}_F+\lambda\mathbbm{1})^{-1}(A_{\lambda,\alpha}u_\xi)+u_\xi\left|\!
\begin{array}{c}
f\in\mathcal{D}(\mathring{H}) \\
u_\xi\in\mathcal{D}(A_{\lambda,\alpha})
\end{array}\!\!\!\right.\right\}\,.
\end{equation}
Indeed, in comparison with the general formula \eqref{eq:ST}, $\mathcal{D}(A_{\lambda,\alpha})$ has a trivial orthogonal complement in $\ker (\mathring{H}^*+\lambda\mathbbm{1})$. 
For each $g\in\mathcal{D}(\mathring{H}_{\langle\alpha\rangle})$ one deduces from \eqref{eq:actionHdotstar_to_g-f_2+1} that
\begin{equation}\label{eq:action_Halpha_2+1}
((\mathring{H}_{\langle\alpha\rangle}+\lambda\mathbbm{1})g)^{\widehat{\;}}(p_1,p_2)\;=\;(p_1^2+p_2^2+\mu p_1\cdot p_2+\lambda)\widehat{f}(p_1,p_2)+\widehat{(A_{\lambda,\alpha}u_\xi)}(p_1,p_2)\,.
\end{equation}

A comparison between \eqref{eq:D_Halpha_2+1} and \eqref{eq:decompositionD_Hdostar_2+1} shows that $\mathcal{D}(\mathring{H}_{\langle\alpha\rangle})$ is obtained as a restriction $\mathcal{D}(\mathring{H}^*)$ by imposing the condition
\begin{equation}\label{eq:TMS_2+1_a}
u_\eta\;=\;A_{\lambda,\alpha} u_\xi
\end{equation}
as an identity in $\ker (\mathring{H}^*+\lambda\mathbbm{1})$ which, by the unitary equivalence \eqref{eq:unitary_equiv_A_Atilde_2+1}, is equivalent to
\begin{equation}\label{eq:TMS_2+1_b}
\eta\;=\;\widetilde{\mathcal{A}}_{\lambda,\alpha}\,\xi
\end{equation}
as an identity in $H^{-1/2}_{W_\lambda}(\mathbb{R}^3)$. The self-adjoint extension $\mathring{H}_{\langle\alpha\rangle}$ is given by the restriction of $\mathring{H}^*$ to those elements of $\mathcal{D}(\mathring{H}^*)$ whose charges, in terms of the decomposition \eqref{eq:decompositionD_Hdostar_2+1}, instead of belonging generically to $H^{-1/2}(\mathbb{R}^3)$ are such that $\xi$ belongs to the domain of $\widetilde{\mathcal{A}}_{\lambda,\alpha}$ and $\eta$ is of the form $\widetilde{\mathcal{A}}_{\lambda,\alpha}\xi$.

We now see that $\mathring{H}_{\langle\alpha\rangle}$  is an operator of TMS type. Indeed, owing to the definition \eqref{eq:Atilde_l_a}, condition \eqref{eq:TMS_2+1_b} implies
\begin{equation}\label{eq:TMS_2+1_c}
\alpha\,\xi\;=\;-\mathcal{T}_\lambda\xi +\!\!\!\begin{array}{c}\frac{1}{2}\end{array}\!\!\!W_\lambda\eta\qquad\qquad\forall\xi\in\mathcal{D}(\mathcal{T}_\lambda)
\end{equation}
as an identity in $L^2(\mathbb{R}^3)$. 
In turn, owing to the definition \eqref{eq:defTtildeLambda}, \eqref{eq:TMS_2+1_c} reads
\begin{eqnarray}
\alpha\,\xi\;=\;-T^+_\lambda\xi +\!\!\!\begin{array}{c}\frac{1}{2}\end{array}\!\!\!W_\lambda\eta\,,\qquad & &\xi\in\mathcal{D}(T^+_\lambda)\subset L^2_+(\mathbb{R}^3) \label{eq:TMS_2+1_c_1}\\
\alpha\,\xi\;=\;-S_0\xi +\!\!\!\begin{array}{c}\frac{1}{2}\end{array}\!\!\!W_\lambda\eta\,,\qquad & &\xi\in\mathcal{D}(S_0)\subset L^2_{\ell=0}(\mathbb{R}^3)\,. \label{eq:TMS_2+1_c_2}
\end{eqnarray}
Plugging \eqref{eq:TMS_2+1_c_1}-\eqref{eq:TMS_2+1_c_2} into \eqref{eq:g_*_asymptotics_2+1} yields the following asymptotics for elements in $\mathcal{D}(\mathring{H}_{\langle\alpha\rangle})$ as $R\to\infty$:
\begin{equation}\label{eq:TMS_cond_asymptotics_2+1}
\begin{split}
\int_{\substack{ \\ \,p_2\in\mathbb{R}^3 \\ \! |p_2|<R}}{\:\widehat{g}(p_1,p_2) \,\ud p_2}\;&=\;\widehat{\xi}(p_1)(4\pi R+\alpha)+o(1) \qquad\qquad (\,\xi\in\mathcal{D}(T^+_\lambda)\,)
\end{split}
\end{equation}
\begin{equation}\label{eq:TMS_cond_asymptotics_2+1_l0}
\begin{split}
\int_{\substack{ \\ \,p_2\in\mathbb{R}^3 \\ \! |p_2|<R}}{\:\widehat{g}(p_1,p_2) \,\ud p_2}\;&=\;\widehat{\xi}(p_1)(4\pi R+\alpha)+(\widehat{S_0\xi})(p_1)-(\widehat{T_\lambda\xi})(p_1)+o(1) \\
&\qquad\qquad\qquad\qquad\qquad\qquad\qquad\qquad\! (\,\xi\in \mathcal{D}(S_0)\,)\,.
\end{split}
\end{equation}
The asymptotics \eqref{eq:TMS_cond_asymptotics_2+1} is completely analogous to \eqref{eq:TMS_cond_asymptotics_1} for the two-body system, except for the fact that it only holds for non-spherically symmetric charges. If the charge $\xi$ is spherically symmetric, then \eqref{eq:TMS_cond_asymptotics_2+1} is modified with an additional $O(1)$-term in $R$ as done in \eqref{eq:TMS_cond_asymptotics_2+1_l0}. In fact, both \eqref{eq:TMS_cond_asymptotics_2+1} and \eqref{eq:TMS_cond_asymptotics_2+1_l0} express the same kind of ultra-violet asymptotics at the coincidence hyperplanes, for their formal inverse Fourier transformed version reads
\[
g(x_1,x_2)\;\sim\;\xi(x_1)\Big(\frac{1}{|x_2|}+\alpha\Big)+\chi(x_1)\qquad\textrm{ as }x_2\to 0
\]
where $\chi$ is identically zero if the charge $\xi$ associated to $g\in\mathcal{D}(\mathring{H}_{\langle\alpha\rangle})$ belongs to $H^{3/2}(\mathbb{R}^3)\cap L^2_{+}(\mathbb{R}^3)$, whereas $\widehat{\chi}=\widehat{S_0\xi}-\widehat{T_\lambda\xi}$ if $\xi\in\mathcal{D}(S_0)\subset L^2_{\ell=0}(\mathbb{R}^3)$. In either case the leading singularity in $g(x_1,x_2)$ as $x_2\to 0$ is precisely of the form $|x_2|^{-1}$.

We summarise the above analysis by saying that
\begin{itemize}
 \item the  operator $\mathring{H}_{\langle\alpha\rangle}$, defined by \eqref{eq:D_Halpha_2+1}-\eqref{eq:action_Halpha_2+1} in terms of a self-adjoint extension $\widetilde{\mathcal{A}}_{\lambda,\alpha}$ of $\mathcal{A}_{\lambda,\alpha}=2\,W_\lambda^{-1}(\mathcal{T}_\lambda+\alpha\mathbbm{1})$ on  $H^{-1/2}_{W_\lambda}(\mathbb{R}^3)$, where $\mathcal{T}_\lambda=S_0\oplus T^+_\lambda$, $S_0$ is a densely defined and symmetric operator $S_0$ on $L^2_{\ell=0}(\mathbb{R}^3)$  with $\ran S_0\subset H^{1/2}(\mathbb{R}^3)$, and $T^+_\lambda$ is the component of $T_\lambda$ on $L^2_+(\mathbb{R}^3)\bigoplus_{\ell=1}^\infty L^2_{\ell}(\mathbb{R}^3)$ with domain $H^{3/2}(\mathbb{R}^3)\cap L^2_+(\mathbb{R}^3)$, is a self-adjoint extension  of $\mathring{H}$ because it satisfies the condition \eqref{eq:TMS_2+1_a}, which is a special case of the conditions of self-adjoint extension of the  Kre{\u\i}n-Vi\v{s}ik-Birman theory;
 \item and moreover $\mathring{H}_{\langle\alpha\rangle}$ is a Hamiltonian of Ter-Martirosyan--Skornyakov type because \eqref{eq:TMS_2+1_a}, through \eqref{eq:TMS_2+1_c_1}-\eqref{eq:TMS_2+1_c_2}, \emph{implies} the TMS asymptotics \eqref{eq:TMS_cond_asymptotics_2+1}-\eqref{eq:TMS_cond_asymptotics_2+1_l0}.
\end{itemize}

We observe, however, two fundamental differences with respect to the construction of the point interaction Hamiltonian for the two-body system. 
In the two-body case,  the TMS condition \eqref{eq:TMS_cond_asymptotics_1}, imposed in the asymptotics of the elements of $\mathcal{D}(\mathring{H}^*)$, turns out to be a \emph{condition of self-adjoint extension}. In the 2+1 fermionic system, instead, 
\begin{itemize}
 \item[1.] the TMS condition \eqref{eq:TMS_cond_asymptotics_2+1} is found to hold only for a class of self-adjoint extensions of $\mathring{H}$, those identified a la Kre{\u\i}n-Vi\v{s}ik-Birman by the \emph{functional} constraint \eqref{eq:TMS_2+1_a},
 \item[2.] and if one imposes \eqref{eq:TMS_cond_asymptotics_2+1} as a point-wise identity valid for a generic class of $\xi$'s, one is \emph{not} guaranteed yet  to have identified a domain of self-adjointness for an extension of $\mathring{H}$, because in general such an Ansatz is not implemented by a self-adjoint operator on $\ker(\mathring{H}^*+\lambda\mathbbm{1})$.
\end{itemize}

\begin{remark}{Remark}
A clarification in retrospective on the emergence and the meaning of the operators $S_0$ and $\mathcal{T}_\lambda$ introduced in Proposition \ref{lem:selfadj_hierarchy} above is surely beneficial at this point. The key issue that arises evidently in our discussion is that \eqref{eq:Tlambda} defines a densely defined, symmetric operator $T_\lambda$ on $L^2(\mathbb{R}^3)$ which fails to map a linear space of certain spherically symmetric $H^{\frac{3}{2}}$-functions into $H^{\frac{1}{2}}(\mathbb{R}^3)$, whereas it does map into the latter space all the $H^{\frac{3}{2}}$-functions that are $L^2$-orthogonal to the spherically symmetric ones. As a consequence, unlike what is customarily given for granted at this point in the literature, there arises the issue to make the symmetric operator $W_\lambda^{-1}T_\lambda$ \emph{densely defined} on the space of charges $H^{-1/2}_{W_\lambda}(\mathbb{R}^3)$ and to realise it self-adjointly. However, one would like to have a self-adjoint operator on $H^{-1/2}_{W_\lambda}(\mathbb{R}^3)$ precisely of the form $W_\lambda^{-1}T_\lambda$: indeed on the one hand by general facts (Krein-Vishik-Birman) a self-adjoint operator on $H^{-1/2}_{W_\lambda}(\mathbb{R}^3)$ identifies a self-adjoint Hamiltonian of point interaction for the three-body system, and on the other hand the structure $W_\lambda^{-1}T_\lambda$ results in a Ter-Martirosyan--Skornyakov condition for the elements in the domain of such an Hamiltonian. The way to retrieve a TMS-like Hamiltonian was then to cure $W_\lambda^{-1}T_\lambda$ on the sector of spherical symmetry, by replacing $T_\lambda$ with a modified operator $\mathcal{T}_\lambda=S_0\oplus T_\lambda^+$, where $T_\lambda^+$ is the restriction of $T_\lambda$ to the subspace $\bigoplus_{\ell=1}^\infty L^2_{\ell}(\mathbb{R}^3)$ and $S_0$ is an arbitrary self-adjoint operator in $L^2_{\ell=0}(\mathbb{R}^3)$ \emph{with values in} $H^{\frac{1}{2}}(\mathbb{R}^3)$. The corresponding $W_\lambda^{-1}\mathcal{T}_\lambda$ is now symmetric and densely defined on $H^{-1/2}_{W_\lambda}(\mathbb{R}^3)$ and any its self-adjoint extension identifies a TMS Hamiltonian for the three-body systems. For such a Hamiltonian, the TMS condition \eqref{eq:TMS_2+1_c_1} emerges only when the charge $\xi$ has symmetry $\ell\geqslant 1$; on the sector $\ell=0$ what holds instead is the weaker condition \eqref{eq:TMS_2+1_c_2}. The latter still prescribes that the generic element $g(x_1,x_2)$ of the domain of the TMS Hamiltonian has a leading singularity $|x_j|^{-1}$ as $|x_j|\to 0$ when the charge $\xi$ associated to $g$ is spherically symmetric, however no $\alpha$-constraint is prescribed between singular and regular part of $g$. This larger freedom in the unconstrained regular part reflects precisely the arbitrariness of $S_0$.
\end{remark}

\section{Applications and concluding remarks}\label{sec:developments}

For the two-body system with point interaction, imposing the TMS asymptotics at scattering length $(-4\pi\alpha)^{-1}$ selects the whole one-parameter family of self-adjoint extensions of $\mathring{H}$, the formal free Hamiltonian defined away from the coincidence configurations. For larger systems, $\mathring{H}$ has infinite deficiency indices and the TMS asymptotics emerge for a proper subclass of extensions of $\mathring{H}$, provided that the charges are taken in a domain of suitable regularity and symmetry. Indeed, except for the two-body case, the TMS condition expresses \emph{point-wise} asymptotics, which per se is not enough to be a condition of self-adjointness: the latter has to be a suitable \emph{functional} condition, such as \eqref{eq:TMS_2+1_a} in the preceding discussion.

Recognising a TMS condition as a self-adjointness condition, by means of the general classification of self-adjoint extensions given by the Kre{\u\i}n-Vi\v{s}ik-Birman theory, is an idea that dates back to the original announcements \cite{Minlos-Faddeev-1961-1,Minlos-Faddeev-1961-2} by Minlos and Faddeev in 1961 and it has been exploited in a series of works by Minlos and collaborators \cite{Minlos-1987,Minlos-Shermatov-1989,Menlikov-Minlos-1991,Menlikov-Minlos-1991-bis,Minlos-TS-1994,Shermatov-2003,Minlos-2011-preprint_May_2010,Minlos-2010-bis,Minlos-2012-preprint_30sett2011,Minlos-2012-preprint_1nov2012,Minlos-RusMathSurv-2014}.

To our understanding, however, the issue of making the operator that in our notation reads $W_\lambda^{-1}T_\lambda$  a well-defined map on the space of charges $H^{-1/2}(\mathbb{R}^3)$, more precisely the issue on whether $\ran \,T_\lambda\subset H^{1/2}(\mathbb{R}^3)=\ran W_\lambda$, was never addressed, nor was it noted that $T_\lambda$ fails to map spherically symmetric functions of $H^{3/2}(\mathbb{R}^3)$ into $H^{1/2}(\mathbb{R}^3)$ (see Proposition \ref{prop:T-W}(i) and its proof). In fact, in all recent works \cite{Minlos-2011-preprint_May_2010,Minlos-2010-bis,Minlos-2012-preprint_30sett2011,Minlos-2012-preprint_1nov2012,Minlos-RusMathSurv-2014} the initial domain of $T_\lambda$ is taken to be $H^1(\mathbb{R}^3)$ and hence, owing to Proposition \ref{prop:T-W}(i), in general $\ran \,T_\lambda$ consists of $L^2$-functions that cannot be pulled back to $H^{-\frac{1}{2}}$-functions by the inverse of $W_\lambda$.

Also, the ubiquitous statement in the above-mentioned literature, according to which each self-adjoint realisation in $L^2(\mathbb{R}^3)$ of the operator $T_\lambda$ defined formally by \eqref{eq:Tlambda} identifies (by general facts of the Kre{\u\i}n-Vi\v{s}ik-Birman theory) a self-adjoint extension of $\mathring{H}$ that displays the TMS asymptotics for the functions of its domain, need be made more precise in two crucial respects. First, one has to factor out the part of $T_\lambda$ that acts on $L^2_{\ell=0}(\mathbb{R}^3)$, the spherically symmetric functions, as we argued in Proposition \ref{lem:selfadj_hierarchy}  and in the discussion that followed from it. Second, even when $T_\lambda$ is self-adjoint on the subspace of higher momentum charges, the corresponding densely defined and symmetric operator $\mathcal{A}_{\lambda,\alpha}=2\,W_\lambda^{-1}(\mathcal{T}_\lambda+\alpha\mathbbm{1})$ on $H^{-1/2}_{W_\lambda}(\mathbb{R}^3)$ is not necessarily self-adjoint and may in turn admit a multiplicity of self-adjoint extensions, as we observed in Remark \ref{rem:selfadj_of_T_not_enough}: only a self-adjoint extension $\widetilde{\mathcal{A}}_{\lambda,\alpha}$ of $\mathcal{A}_{\lambda,\alpha}$ identifies, a la Kre{\u\i}n-Vi\v{s}ik-Birman, the self-adjoint extension $\mathring{H}_{\langle\alpha\rangle}$ of $\mathring{H}$.

It remains therefore unclear how to relate the range of masses $m$ (identified in \cite{Minlos-2011-preprint_May_2010,Minlos-2010-bis,Minlos-2012-preprint_30sett2011,Minlos-2012-preprint_1nov2012,Minlos-RusMathSurv-2014}) in which $T_\lambda$, initially defined on $H^{1}(\mathbb{R}^3)$, is self-adjoint or has a family of self-adjoint extensions on $L^2(\mathbb{R}^3)$ 
with the actual range of masses in which $\mathring{H}$ admits one or more self-adjoint extensions displaying the TMS asymptotics.
The fact that $\mathcal{A}_{\lambda,\alpha}=2\,W_\lambda^{-1}(\mathcal{T}_\lambda+\alpha\mathbbm{1})$ may be only symmetric on $H^{-1/2}_{W_\lambda}(\mathbb{R}^3)$ even when $\mathcal{T}_\lambda$ is self-adjoint on $L^2(\mathbb{R}^3)$ should account for a \emph{larger} range of masses in which $\mathring{H}$ has a multiplicity of TMS-like self-adjoint extensions than the range in which $T_\lambda$ has a multiplicity of self-adjoint extensions.


It becomes of great interest now to re-read and understand, in terms of the general classification of self-adjoint extensions of $\mathring{H}$ provided by the Kre{\u\i}n-Vi\v{s}ik-Birman theory, those results that produced TMS self-adjoint extensions of $\mathring{H}$ through an approach based on quadratic forms on Hilbert space (see 
\cite{Teta-1989,dft-Nparticles-delta,DFT-proc1995,Finco-Teta-2012,CDFMT-2012} and above all the recent work \cite{CDFMT-2015}). Indeed, the quadratic form approach produces a single self-adjoint TMS Hamiltonian, or alternatively a family of self-adjoint TMS Hamiltonians, all extensions of  $\mathring{H}$, in a regime of masses that differs from what is known from the operator-theoretic approach and that coincides instead to what is found in the physical literature -- see \cite[Remark 2.4 and Proposition 2.2]{CDFMT-2015}, as well as the discussion around \cite[Eq.~(1.17)]{CDFMT-2015}. Owing to the general picture of the self-adjoint extension theory, each such Hamiltonian \emph{must} be selected by a condition on the charges realised by a self-adjoint map on $\ker(\mathring{H}^*+\lambda\mathbbm{1})$ as in \eqref{eq:TMS_2+1_a} or, more generally, in \eqref{eq:SBsmbb-iff-invBsmbb_Tversion}. One should identify such a map and to compare it to its analog in the operator-theoretic approach.

Armed with the analysis and the discussion developed here, we plan to address these issues in a follow-up work.

\appendix

\section{Basics of the Kre{\u\i}n-Vi\v{s}ik-Birman self-adjoint extension theory}\label{app:KVB}

In this appendix we collect the main results of the Kre{\u\i}n-Vi\v{s}ik-Birman theory of self-adjoint extensions of semi-bounded symmetric operators. This theory was developed by  Kre{\u\i}n \cite{Krein-1947}, Vi\v{s}ik \cite{Vishik-1952}, and Birman \cite{Birman-1956}) between the mid 1940's and the mid 1950's. For the present formulation and the proof of all the statements that follow we refer to the comprehensive discussion \cite{M-KVB2015}, as well as to the expository works \cite{Flamand-Cargese1965,Alonso-Simon-1980}.

For any given symmetric operator $S$ with domain $\mathcal{D}(S)$, let 
\begin{equation}
m(S)\;:=\;\inf_{\substack{f\in\mathcal{D}(S) \\ f\neq 0}}\frac{\langle f,Sf\rangle}{\|f\|^2}\,.
\end{equation}
be the ``bottom'' of $S$, i.e., its greatest lower bound. Hereafter $S$ shall be semi-bounded below, meaning therefore $m(S)>-\infty$. It is not restrictive to assume henceforth
\begin{equation}
m(S)>0\,,
\end{equation}
for in the general case one applies the discussion that follows to the strictly positive operator $S+\lambda\mathbbm{1}$, $\lambda>-m(S)$, and  then re-express trivially the final results in terms of the original $S$. When $S$ is densely defined, \emph{the choice $m(S)>0$ implies that the Friedrichs extension  $S_F$ of $S$ is invertible with bounded inverse defined everywhere on $\cH$}: this will allow $S_F^{-1}$ to enter directly the discussion. In the general case in which $S_F$ is not necessarily invertible, the role of $S_F^{-1}$ is naturally replaced by the inverse $\widetilde{S}^{-1}$ of any a priori known self-adjoint extension $\widetilde{S}$ of $S$, which thus takes the role of given ``datum'' of the theory. Moreover, with the choice $m(S)>0$, the level $0$ becomes naturally the reference value with respect to which to express the other distinguished (canonically given) extension of $S$, the  Kre{\u\i}n-von Neumann extension $S_N$.

\begin{lemma}{Lemma}\label{lemma:krein_decomp_formula} \emph{(Decomposition formulas)} For a densely defined symmetric operator $S$ with positive bottom, one has
\begin{eqnarray}
\mathcal{D}(S^*) \!\!\!& = &\!\!\! \mathcal{D}(S_F)\dotplus \ker S^* \label{eq:DomS*krein} \\
\mathcal{D}(S^*)  \!\!\!& = &\!\!\! \mathcal{D}(\overline{S})\dotplus S_F^{-1} \ker S^*\dotplus \ker S^*  \label{eq:DomS*}\\
\mathcal{D}(S_F)  \!\!& = &\!\!\! \mathcal{D}(\overline{S})\dotplus S_F^{-1} \ker S^*\,. \label{eq:DomSF}
\end{eqnarray}
\end{lemma}

\begin{theorem}{Theorem}\label{thm:VB-representaton-theorem_Tversion} \emph{(Classification of self-adjoint extensions -- operator version.)}
Let $S$ be a densely defined symmetric operator on a Hilbert space $\cH$ with positive bottom ($m(S)>0$). 
There is a one-to-one correspondence between the  family of all self-adjoint extensions of  $S$ on $\cH$ and the family of the self-adjoint operators on Hilbert subspaces of $\ker S^*$.
If $T$ is any such operator, in the correspondence $T\leftrightarrow S_T$ each self-adjoint extension $S_T$ of $S$ is given by
\begin{equation}\label{eq:ST}
\begin{split}
S_T\;&=\;S^*\upharpoonright\mathcal{D}(S_T) \\
\mathcal{D}(S_T)\;&=\;\left\{f+S_F^{-1}(Tv+w)+v\left|\!\!
\begin{array}{c}
f\in\mathcal{D}(\overline{S})\,,\;v\in\mathcal{D}(T) \\
w\in\ker S^*\cap\mathcal{D}(T)^\perp
\end{array}\!\!
\right.\right\}.
\end{split}
\end{equation}
\end{theorem}

\begin{theorem}{Theorem}\label{thm:semibdd_exts_operator_formulation_Tversion}\emph{(Characterisation of semi-bounded extensions.)}
Let $S$ be a densely defined symmetric operator on a Hilbert space $\cH$ with positive bottom. 
If, with respect to the notation of \eqref{eq:ST}, $S_T$ is a self-adjoint extension of $S$, and if $\alpha<m(S)$, then
\begin{equation}\label{eq:SBsmbb-iff-invBsmbb_Tversion}
\begin{split}
\langle g,S_T g\rangle\;&\geqslant\;\alpha\,\|g\|^2\qquad\forall g\in\mathcal{D}(S_T) \\
& \Updownarrow \\
\langle v,T v\rangle\;\geqslant\;\alpha\|v\|^2+\:&\alpha^2\langle v,(S_F-\alpha\mathbbm{1})^{-1} v\rangle\qquad\forall v\in\mathcal{D}(T)\,.
\end{split}
\end{equation}
As an immediate consequence, $m(T)\geqslant m(S_T)$ for any semi-bounded $S_T$.
In particular, positivity or strict positivity of the bottom of $S_T$ is equivalent to the same property for $T$, that is,
 \begin{equation}\label{eq:positiveSBiffpositveB-1_Tversion}
 \begin{split}
 m(S_T)\;\geqslant \;0\quad&\Leftrightarrow\quad m(T)\;\geqslant\; 0 \\
 m(S_T)\;> \;0\quad&\Leftrightarrow\quad m(T)\;>\; 0\,.
 \end{split}
 \end{equation}
Moreover, if $m(T)>-m(S)$, then
 \begin{equation}\label{eq:bounds_mS_mB_Tversion}
 m(T)\;\geqslant\; m(S_T)\;\geqslant\;\frac{m(S) \,m(T)}{m(S)+m(T)}\,.
 \end{equation}
\end{theorem}

\begin{theorem}{Theorem}\label{thm:semibdd_exts_form_formulation_Tversion}\emph{(Characterisation of semi-bounded extensions -- form version.)}
Let $S$ be a densely defined symmetric operator on a Hilbert space $\cH$ with positive bottom and, with respect to the notation of \eqref{eq:ST}, let $S_T$ be a  \emph{semi-bounded} (not necessarily positive) self-adjoint extension of $S$. Then
\begin{equation}\label{eq:D[SB]_Tversion}
\mathcal{D}[T]\;=\; \mathcal{D}[S_T]\,\cap\,\ker S^*
\end{equation}
and
 \begin{equation}\label{eq:decomposition_of_form_domains_Tversion}
 \begin{split}
 \mathcal{D}[S_T]\;&=\;\mathcal{D}[S_F]\,\dotplus\,\mathcal{D}[T] \\
 S_T[f+v,f'+v']\;&=\;S_F[f,f']\,+\,T[v,v'] \\
 &\forall f,f'\in\mathcal{D}[S_F],\;\forall v,v'\in\mathcal{D}[T]\,.
 \end{split}
\end{equation}
As a consequence,
\begin{equation}\label{eq:extension_ordering_Tversion}
S_{T_1}\,\geqslant\,S_{T_2}\qquad\Leftrightarrow\qquad T_1\,\geqslant\,T_2
\end{equation}
and
\begin{equation}
T\;\geqslant\;S_T\,.
\end{equation}
\end{theorem}

\begin{theorem}{Proposition}\label{prop:parametrisation_SF_SN_Tversion}\emph{(Parametrisation of $S_F$ and $S_N$.)}
Let $S$ be a densely defined symmetric operator on a Hilbert space $\cH$ with positive bottom and let $S_T$ be a \emph{positive} self-adjoint extension of $S$, parametrised by $T$  according to Theorems \ref{thm:VB-representaton-theorem_Tversion} and \ref{thm:semibdd_exts_form_formulation_Tversion}.
\begin{itemize}
 \item[(i)] $S_T$ is the \emph{Friedrichs extension} when $\mathcal{D}[T]=\{0\}$ (`` $T=\infty$'').
 \item[(ii)] $S_T$ is the \emph{Kre{\u\i}n-von Neumann extension} when $\mathcal{D}(T)=\mathcal{D}[T]=\ker S^*$ and $Tu=0$ $\forall u\in\ker S^*$ (	$T=\mathbb{O}$).
\end{itemize}
\end{theorem}

\begin{theorem}{Proposition}\label{cor:finite_deficiency_index}\emph{(Finite deficiency index.)}
If $S$ is a semi-bounded and densely defined symmetric operator on a Hilbert space $\cH$ with finite deficiency index, then 
\begin{itemize}
 \item[(i)] the semi-boundedness of $S_T$ is equivalent to the semi-boundedness of $T$;
 \item[(ii)] any self-adjoint extension of $S$ is bounded below.
\end{itemize}
\end{theorem}

\begin{theorem}{Proposition}\label{cor:finite-dimensional}\emph{(``Finite-dimensional'' extensions are always semi-bounded.)}
Given a semi-bounded and densely defined symmetric operator $S$ on a Hilbert space $\cH$, whose bottom is positive, all the self-adjoint extensions of $S_T$ of $S$ for which the parameter $T$, in the parametrisation \eqref{eq:ST} of Theorem \ref{thm:VB-representaton-theorem_Tversion}, is a self-adjoint operator acting on a \emph{finite-dimensional} subspace of $\ker S^*$ are semi-bounded. For the occurrence of unbounded below self-adjoint extensions it is necessary that $\dim\overline{\mathcal{D}(T)}=\infty$.
\end{theorem}

\section{Proof of the identity \eqref{eq:H20approx}}\label{app:approx}

We prove here that the $H^2$-closure of $C^\infty_0(\mathbb{R}^3\!\setminus\!\{0\})$, which is by definition the space $H^2_0(\mathbb{R}^3\!\setminus\!\{0\})$, coincides with the space $\{ f\in H^2(\mathbb{R}^3)\,|\, f(0)=0\}$, thus obtaining the identity \eqref{eq:H20approx}. Clearly 
\begin{equation*}
\overline{\,C^\infty_0(\mathbb{R}^3\!\setminus\!\{0\})\,}^{\|\,\|_{H^2}}\;\subset\;\big\{ f\in H^2(\mathbb{R}^3)\,\big|\, f(0)=0\big\}
\end{equation*}
because the $H^2$-convergence implies the point-wise convergence of continuous functions. Thus, given $f\in H^2(\mathbb{R}^3)$ with $f(0)=0$, we only need to find for arbitrary $\varepsilon>0$ a function $f_\varepsilon\in C^\infty_0(\mathbb{R}^3\!\setminus\!\{0\})$ such that
\begin{equation}\label{eq:approx_to_prove}
\|f-f_\varepsilon\|_{H^2}\;\leqslant\;\varepsilon\,.
\end{equation}

Given a cut-off function $\chi\in C^\infty([0,+\infty))$ such that 
\[
\begin{split}
\chi(r)\;=\;0\qquad & \textrm{ for }\;r\in[0,1]  \\
\chi(r)\;=\;1\qquad & \textrm{ for }\;r\in[2,+\infty)\,,
\end{split}
\]
set
\[
\phi_n(x)\;:=\;\chi(n|x|)\,,\qquad\quad n\in\mathbb{N}\,,\qquad x\in\mathbb{R}^3\,.
\]
Then, for any $n\in\mathbb{N}$, $\phi_n\in C^\infty (\mathbb{R}^3)$ and
\begin{equation}\label{eq:properties_phi-0}
\begin{split}
\phi_n(x)\;=\;0\qquad\qquad & \textrm{ for }\;|x|\;\leqslant\;\frac{1}{n} \\
\phi_n(x)\;=\;1\qquad\qquad & \textrm{ for }\;|x|\;\geqslant\;\frac{2}{n} \\
|\phi_n(x)|\;\leqslant\;c_\chi \qquad\quad\;\; & \\
|\nabla\phi_n(x)|\;\leqslant\;n \,c_\chi\qquad\;\;\; & \;\;\forall x\in\mathbb{R}^3 \\
|\Delta\phi_n(x)|\;\leqslant\;n^2 c_\chi\qquad\;\;& 
\end{split}
\end{equation}
where $c_\chi$ depends only on $\|\chi\|_{\sup}$, $\|\chi'\|_{\sup}$, and $\|\chi''\|_{\sup}$.

Correspondingly, each function $\phi_nf$
belongs to $H^2(\mathbb{R}^3)$ and vanishes when $|x|\leqslant n^{-1}$. Furthermore, we now show that
\begin{equation}\label{eq:H2_approx_fphi_n}
\|\phi_n\,f-f\|_{H^2}\;\xrightarrow[]{\;\; n\to +\infty\;\;}\;0\,.
\end{equation}
Indeed, $\|\phi_n\,f-f\|_{L^2}\to 0$ follows immediately by dominated convergence and \eqref{eq:properties_phi-0}, whereas $\|\Delta(\phi_n\,f)-\Delta f\|_{L^2}\to 0$ is obtained with the following argument.
First, one estimates
\begin{equation*}
\begin{split}
\|\Delta(\phi_n\,f)-\Delta f\|_{L^2}\;\leqslant\;\|(\phi_n-1)\Delta f\|_{L^2}+2\|(\nabla\phi_n)\cdot(\nabla f)\|_{L^2}+\|f\Delta\phi_n\|_{L^2}\,.
\end{split}
\end{equation*}
For the first summand in the r.h.s.~of the inequality above one has $\|(\phi_n-1)\Delta f\|_2\to 0$ as $n\to\infty$ again by dominated convergence and \eqref{eq:properties_phi-0}. 
For the second summand,
\[
\begin{split}
\|(\nabla\phi_n)\cdot(\nabla f)\|^2_{L^2}\;&\leqslant\;n^2\,c_\chi^2\int_{|x|\leqslant 2 n^{-1}}|\nabla f(x)|^2\,\ud x \\
&\leqslant\;n^2\,c_\chi^2\,\Big(\frac{4\pi\,2^3}{3\,n^3}\Big)^{2/3}\|\nabla f\|^2_{L^6(B_{2/n})} \\
&\lesssim\;\|f\|^2_{H^2(B_{2/n})}\;\xrightarrow[]{\;\; n\to +\infty\;\;}\;0
\end{split}
\]
where we used \eqref{eq:properties_phi-0}  in the first step, a H\"{o}lder inequality in the second, the Sobolev embedding in the third, and dominated convergence in the last, $B_r$ denoting the closed ball of $\mathbb{R}^3$ of radius $r$ centred at the origin.
For the third summand,
\[
\begin{split}
\|f\Delta\phi_n\|_{L^2(\mathbb{R}^3)}^2\;&\leqslant\;n^4\,c_\chi^2\int_{\mathrm{supp}(\Delta\phi_n)}|f(x)|^2\,\ud x \\
&\lesssim\;n^3\,\| f\|^2_{H^2(B_{2/n})}\,|B_{2/n}|
\\
&\lesssim\;\|f\|^2_{H^2(B_{2/n})}\;\xrightarrow[]{\;\; n\to +\infty\;\;}\;0
\end{split}
\]
where in the first step we used \eqref{eq:properties_phi-0}, in the second we used the estimate
\[
|f(x)|\;=\;|f(x)-f(0)|\;\leqslant\;\frac{C}{\sqrt{n\,}\,}\,\|\nabla f\|_{L^6(B_{2/n})}\;\lesssim\;\frac{1}{\sqrt{n\,}\,}\| f\|_{H^2(B_{2/n})}
\]
(where the constant $C$ does not depend on $n$) that follows from Morrey's inequality and the Sobolev embedding (see \cite[Section 5.6.2]{Evans-pde},
and in the last step we used dominated convergence.
Thus, \eqref{eq:H2_approx_fphi_n} is proved.

As a consequence of \eqref{eq:H2_approx_fphi_n} above, for the arbitrary $\varepsilon>0$ fixed at the beginning there is $N_\varepsilon\in\mathbb{N}$ and $\delta:=N_\varepsilon^{-1}$ such that $g_\varepsilon:=\phi_{N_\varepsilon}f$ is a $H^2$-function satisfying
\begin{equation}\label{eq:fapprox-p1}
g_\varepsilon\equiv 0\;\textrm{ on }\;B_\delta\qquad\textrm{and}\qquad\|f-g_\varepsilon\|_{H^2}\;\leqslant\;\frac{\varepsilon}{3}\,.
\end{equation}

We consider now a standard mollification $j_n*g_\varepsilon$ of $g_\varepsilon$ for some $j\in C^\infty_0(\mathbb{R}^3)$ with $\int_{\mathbb{R}^3} j\,\ud x =1$ and $j_n(x):=n^{-3}j(nx)$, $n\in\mathbb{N}$, $x\in\mathbb{R}^3$. Then $j_n*g_\varepsilon\in C^\infty(\mathbb{R}^3)\cap H^2(\mathbb{R}^3)$ and $j_n*g_\varepsilon\to g_\varepsilon$ in $H^2$ as $n\to\infty$. Therefore, there is $n_\varepsilon\in\mathbb{N}$ large enough so that $h_\varepsilon:=j_{n_\varepsilon}*g_\varepsilon$ satisfies
\begin{equation*}
\|g_\varepsilon-h_\varepsilon\|_{H^2}\;\leqslant\;\frac{\varepsilon}{3}
\end{equation*}
\emph{and}  $\mathrm{supp}(j_{n_\varepsilon})\subset B_{\delta/2}$ (where $\delta$ is the radius of the ball that contains the support of $g_\varepsilon$, see \eqref{eq:fapprox-p1} above). As a consequence, $h_\varepsilon$ vanishes in $B_{\delta/2}$: indeed, if $x\in B_{\delta/2}$ and $y\in\mathrm{supp}(j_{n_\varepsilon})$, that is, $|x|\leqslant\frac{\delta}{2}$ and $y\leqslant\frac{\delta}{2}$, one has $|x-y|\leqslant\delta$ and hence
\[
h_\varepsilon(x)\;=\;(j_{n_\varepsilon}*g_\varepsilon)(x)\;=\;\int_{\mathrm{supp}(j_{n_\varepsilon})}g_\varepsilon(x-y)\,j_{n_\varepsilon}(y)\,\ud y\;=\;0
\]
because $g_\varepsilon$ vanishes in $B_\delta$. Summarising, we have found $h_\varepsilon\in C^\infty(\mathbb{R}^3)\cap H^2(\mathbb{R}^3)$ such that
\begin{equation}\label{eq:fapprox-p2}
h_\varepsilon\equiv 0\;\textrm{ on }\;B_{\delta/2}\qquad\textrm{and}\qquad\|g_\varepsilon-h_\varepsilon\|_{H^2}\;\leqslant\;\frac{\varepsilon}{3}\,.
\end{equation}

Last, we consider a cut-off function $\zeta\in C^\infty_0([0,+\infty))$ such that 
\[
\begin{split}
\zeta(r)\;=\;1\qquad & \textrm{ for }\;r\in[0,1]  \\
\zeta(r)\;=\;0\qquad & \textrm{ for }\;r\in[2,+\infty)\,,
\end{split}
\]
and set 
\[
h_{\varepsilon,n}:=\;\zeta(n^{-1}x)h_\varepsilon(x)\qquad\quad n\in\mathbb{N}\,,\qquad x\in\mathbb{R}^3\,.
\]
Then, in complete analogy to the reasoning above, we see that for $n$ large enough each $h_{\varepsilon,n}$ belongs to $C^\infty_0(\mathbb{R}^3\!\setminus\!B_{\delta/2})\subset C^\infty_0(\mathbb{R}^3\!\setminus\!\{0\})$ and $h_{\varepsilon,n}\to h_\varepsilon$ in $H^2$ as $n\to\infty$. This implies the existence of $M_\varepsilon\in\mathbb{N}$ such that the function $f_\varepsilon:=h_{\varepsilon,M_\varepsilon}$ satisfies
\begin{equation}\label{eq:fapprox-p3}
f_\varepsilon\;\in\;C^\infty_0(\mathbb{R}^3\!\setminus\!\{0\})\qquad\textrm{and}\qquad\|h_\varepsilon-f_\varepsilon\|_{H^2}\;\leqslant\;\frac{\varepsilon}{3}\,.
\end{equation}

Using \eqref{eq:fapprox-p1}, \eqref{eq:fapprox-p2}, and \eqref{eq:fapprox-p3} above in a triangular inequality we finally conclude
\[
\|f-f_\varepsilon\|_{H^2}\;\leqslant\;\|f-g_\varepsilon\|_{H^2}+\|g_\varepsilon-h_\varepsilon\|_{H^2}+\|h_\varepsilon-f_\varepsilon\|_{H^2}\;\leqslant\;\varepsilon
\]
for a function $f_\varepsilon\;\in\;C^\infty_0(\mathbb{R}^3\!\setminus\!\{0\})$, which completes the proof of \eqref{eq:approx_to_prove}.

\section{Proof of the inclusions and the identities in \eqref{useful_inclusion} and \eqref{useful_inclusion_2+1}}\label{app:inclusions}

We prove here \eqref{useful_inclusion_2+1}: the proof applies straightforwardly also to obtain  \eqref{useful_inclusion}.

Let $f\in H^2(\mathbb{R}^3\times\mathbb{R}^3)$. We want to show that $f$ belongs to $H_0^1((\mathbb{R}^3\times\mathbb{R}^3)\!\setminus\!(\Gamma_1\cup\Gamma_2))$. Since
\[
H_0^1((\mathbb{R}^3\times\mathbb{R}^3)\!\setminus\!(\Gamma_1\cup\Gamma_2))\;=\;\overline{\,C_0^\infty((\mathbb{R}^3\times\mathbb{R}^3)\!\setminus\!(\Gamma_1\cup\Gamma_2))\,}^{\,\|\;\|_{H^1}}\,,
\]
it is enough to find, for arbitrary $\varepsilon>0$, a function $f_\varepsilon\in C_0^\infty((\mathbb{R}^3\times\mathbb{R}^3)\!\setminus\!(\Gamma_1\cup\Gamma_2))$ such that 
\begin{equation}\label{eq:triang}
\|f-f_\varepsilon\|_{H^1}\;\leqslant\;\varepsilon\,.
\end{equation}

First, since $f\in H^2(\mathbb{R}^3\times\mathbb{R}^3)\subset H^1(\mathbb{R}^3\times\mathbb{R}^3)$ and $C_0^\infty(\mathbb{R}^3\times\mathbb{R}^3)$ is dense in $H^1(\mathbb{R}^3\times\mathbb{R}^3)$, there exists $g_\varepsilon\in C_0^\infty(\mathbb{R}^3\times\mathbb{R}^3)$ such that 
\begin{equation}\label{eq:triang1}
\|f-g_\varepsilon\|_{H^1}\;\leqslant\;\frac{\varepsilon}{2}\,.
\end{equation}

Given a cut-off function $\chi\in C^\infty([0,+\infty))$ such that 
\[
\begin{split}
\chi(r)\;=\;0\qquad & \textrm{ for }\;r\in[0,1]  \\
\chi(r)\;=\;1\qquad & \textrm{ for }\;r\in[2,+\infty)\,,
\end{split}
\]
set
\[
\phi_n(x,y)\;:=\;\chi(n|x|)\chi(n|y|)\,,\qquad\quad n\in\mathbb{N}\,,\qquad (x,y)\in\mathbb{R}^3\times\mathbb{R}^3\,.
\]
Then, for any $n\in\mathbb{N}$, $\phi_n\in C^\infty (\mathbb{R}^3\times\mathbb{R}^3)$ and
\begin{equation}\label{eq:properties_phi}
\begin{split}
\phi_n(x,y)\;=\;0\qquad\qquad & \textrm{ for }\;|x|\;\leqslant\;\frac{1}{n}\textrm{ or }|y|\;\leqslant\;\frac{1}{n} \\
\phi_n(x,y)\;=\;1\qquad\qquad & \textrm{ for }\;|x|\;\geqslant\;\frac{2}{n}\textrm{ and }|y|\;\geqslant\;\frac{2}{n} \\
|\phi_n(x,y)|\;\leqslant\;c_\chi \qquad\quad\;\; & \;\;\forall (x,y)\in\mathbb{R}^3\times\mathbb{R}^3 \\
|\nabla\phi_n(x,y)|\;\leqslant\;n\,c_\chi\qquad\;\;\: & \;\;\forall (x,y)\in\mathbb{R}^3\times\mathbb{R}^3 \,,
\end{split}
\end{equation}
where here and henceforth $\nabla$ denotes the $6$-dimensional gradient and $c_\chi$ depends only on $\|\chi\|_{\sup}$ and $\|\chi'\|_{\sup}$.

Correspondingly, each function
\[
g_{\varepsilon,n}\;:=\;\phi_n \,g_\varepsilon
\]
belongs to $C^\infty_0((\mathbb{R}^3\times\mathbb{R}^3)\!\setminus\!(\Gamma_1\cup\Gamma_2))$ and we now show that
\begin{equation}\label{eq:H1approx}
\|g_{\varepsilon,n}-g_\varepsilon\|_{H^1}\;\xrightarrow[]{\;\; n\to +\infty\;\;}\;0\,.
\end{equation}
Indeed, $\|g_{\varepsilon,n}-g_\varepsilon\|_{L^2}\to 0$ follows immediately by dominated convergence and \eqref{eq:properties_phi}, whereas $\|\nabla g_{\varepsilon,n}-\nabla g_\varepsilon\|_{L^2}\to 0$ follows from the vanishing of both summands in the r.h.s.~of the  inequality
\[
\|\nabla g_{\varepsilon,n}-\nabla g_\varepsilon\|_{L^2}\;\leqslant\;\|\phi_n\nabla g_{\varepsilon}-\nabla g_\varepsilon\|_{L^2}+\|g_{\varepsilon}\nabla\phi_n\|_{L^2}\,.
\]
Explicitly, $\|\phi_n\nabla g_{\varepsilon}-\nabla g_\varepsilon\|_{L^2}\to 0$ by dominated convergence, owing to \eqref{eq:properties_phi}, whereas
\begin{equation*}
\begin{split}
\|g_{\varepsilon}\nabla\phi_n\|_{L^2(\mathbb{R}^3)}^2\;&\leqslant\;n^2\,c_\chi^2\iint_{\mathrm{supp}(\nabla\phi_n)\cap\mathrm{supp}(g_\varepsilon)} |g_\varepsilon|^2\,\ud x\,\ud y \\
&\leqslant\;n^2\,c_\chi^2\Big(\frac{\,C_{g_\varepsilon}}{n^3}\Big)^{\!\frac{2}{3}}\|g_\varepsilon\|^2_{L^6(B_{\varepsilon,n})} \\
&\lesssim\;\|g_\varepsilon\|^2_{H^2(B_{\varepsilon,n})}\;\xrightarrow[]{\;\; n\to +\infty\;\;}\;0\,,
\end{split}
\end{equation*}
where we used \eqref{eq:properties_phi} in the first step, 
a H\"{o}lder inequality and the estimate
\[
\begin{split}
B_{\varepsilon,n}\;:=&\;\;\mathrm{supp}(\nabla\phi_n)\cap\mathrm{supp}(g_\varepsilon) \\
|B_{\varepsilon,n}|\;\leqslant&\;\;\big|\big(\{|x|\leqslant 2/n\}\cup\{|y|\leqslant 2/n\}\big)\cap\mathrm{supp}(g_\varepsilon)\big|\;\leqslant C_{g_\varepsilon} n^{-3}
\end{split}
\]
in the second step, where $C_{g_\varepsilon}$ depends only on the radius of $\mathrm{supp}(g_\varepsilon)$, the continuous embedding $H^2(B_{\varepsilon,n})\subset L^6(B_{\varepsilon,n})$ in the third step, and dominated convergence in the last step.

As a consequence of \eqref{eq:H1approx} above, for the arbitrary $\varepsilon>0$ fixed at the beginning there is $N_\varepsilon\in\mathbb{N}$ such that $f_\varepsilon:=g_{\varepsilon,N_{\varepsilon}}\in C^\infty_0((\mathbb{R}^3\times\mathbb{R}^3)\!\setminus\!(\Gamma_1\cup\Gamma_2))$ satisfies
\begin{equation}\label{eq:triang2}
\|g_\varepsilon-f_\varepsilon\|_{H^1}\;\leqslant\;\frac{\varepsilon}{2}\,.
\end{equation}
A triangular inequality based on \eqref{eq:triang1} and \eqref{eq:triang2} then yields \eqref{eq:triang}, thus completing the proof of \eqref{useful_inclusion_2+1}.

Clearly, \eqref{eq:H1approx} also shows that to any function in $C^\infty_0(\mathbb{R}^3\times\mathbb{R}^3)$ there is a function in $C^\infty_0((\mathbb{R}^3\times\mathbb{R}^3)\!\setminus\!(\Gamma_1\cup\Gamma_2))$ arbitrarily close in the $H^1$-norm, which implies that the spaces $H^1(\mathbb{R}^3\times\mathbb{R}^3)$ and $H^1_0((\mathbb{R}^3\times\mathbb{R}^3)\!\setminus\!(\Gamma_1\cup\Gamma_2))$ coincide. This yields the identity in \eqref{useful_inclusion_2+1}.

The arguments above 
apply virtually unchanged both for the inclusion $H^2(\mathbb{R}^3)\subset H^1_0(\mathbb{R}^3\!\setminus\!\{0\})$ and for the identity $H^1_0(\mathbb{R}^3\!\setminus\!\{0\})=H^1(\mathbb{R}^3)$, with the obvious removal of one variable, thus proving \eqref{useful_inclusion}.


\section*{Acknowledgements}

\noindent We warmly thank G.~Dell'Antonio and L.~Guerini for enlightening discussions on the subject.


\def\cprime{$'$}

\end{document}